\documentclass[sigconf]{acmart}
\AtBeginDocument{%
  \providecommand\BibTeX{{%
    \normalfont B\kern-0.5em{\scshape i\kern-0.25em b}\kern-0.8em\TeX}}}

\newcommand{\ie}{\emph{i.e., }}
\newcommand{\eg}{\emph{e.g., }}

\newcommand{\wrt}{\emph{w.r.t. }}

\newcommand{\zy}[1]{\textcolor{black}{#1}}
\newcommand{\sth}[1]{\textcolor{black}{#1}}

\usepackage{enumitem}
\usepackage{verbatim}
\usepackage[ruled]{algorithm2e}
\usepackage{multirow}
\usepackage{subfigure} 
\usepackage{ragged2e}
\usepackage{caption}
\usepackage{graphicx}
\usepackage[normalem]{ulem}
\useunder{\uline}{\ul}{}
\usepackage{amsmath}
\usepackage{setspace}
\usepackage{bm}
\newtheorem{theorem}{Theorem}
\newtheorem{definition}{Definition}[section]

\copyrightyear{2024}
\acmYear{2024}
\setcopyright{acmlicensed}\acmConference[SIGIR '24]{Proceedings of the 47th International ACM SIGIR Conference on Research and Development in Information Retrieval}{July 14--18, 2024}{Washington, DC, USA}
\acmBooktitle{Proceedings of the 47th International ACM SIGIR Conference on Research and Development in Information Retrieval (SIGIR '24), July 14--18, 2024, Washington, DC, USA}
\acmDOI{10.1145/3626772.3657822}
\acmISBN{979-8-4007-0431-4/24/07}

\begin{document}

% \title{Stable Feature Interaction Modeling for Click-Through Rate Prediction}
% \title{Stable Feature Interaction Learning for Feature-rich Recommendation}
\title{Fair Recommendations with Limited Sensitive Attributes: A Distributionally Robust Optimization Approach}

\author{Tianhao Shi}
\affiliation{%
  \institution{University of Science and Technology of China}
    \city{Hefei}
  \country{China}}
\email{sth@mail.ustc.edu.cn}

\author{Yang Zhang}
\authornote{Corresponding authors.}
\affiliation{%
  \institution{University of Science and Technology of China}
    \city{Hefei}
  \country{China}}
  \email{zyang1580@gmail.com}
% \email{zy2015@mail.ustc.edu.cn}

\author{Jizhi Zhang}
\affiliation{%
  \institution{University of Science and Technology of China}
    \city{Hefei}
  \country{China}}
\email{cdzhangjizhi@mail.ustc.edu.cn}

\author{Fuli Feng}
\authornotemark[1]
\affiliation{%
  \institution{University of Science and Technology of China}
    \city{Hefei}
  \country{China}}
\email{fulifeng93@gmail.com}

\author{Xiangnan He}
% \authornotemark[2]
\affiliation{%
  \institution{University of Science and Technology of China}
    \city{Hefei}
  \country{China}}
\email{xiangnanhe@gmail.com}

\renewcommand{\shortauthors}{Tianhao Shi, et al.}
\begin{abstract}

As recommender systems are indispensable in various domains such as job searching and e-commerce, providing equitable recommendations to users with different sensitive attributes becomes an imperative requirement.
% Existing approaches to promoting fairness in recommender systems typically assume having complete access to all sensitive attributes, which can be challenging to obtain due to users' privacy concerns or lacking mechanisms for collecting sensitive attributes.
Prior approaches for enhancing fairness in recommender systems presume the availability of all sensitive attributes, which can be difficult to obtain due to privacy concerns or inadequate means of capturing these attributes.
In practice, the efficacy of these approaches is limited, pushing us to investigate ways of promoting fairness with limited sensitive attribute information. 

Toward this goal, it is important to reconstruct missing sensitive attributes.
Nevertheless, reconstruction errors are inevitable due to the complexity of real-world sensitive attribute reconstruction problems and legal regulations. Thus, we pursue fair learning methods that are robust to reconstruction errors.
% Nevertheless, using the reconstructed sensitive attribute directly in fairness-aware methods is not feasible due to the existence of reconstruction error and the potential unwillingness of users to participate in the reconstruction.
% Also, legal regulations regarding user consent may hinder the reconstruction of complete sensitive attributes.Distributionally Robust Optimization for Fairness with Limited Sensitive Attributes
To this end, we propose \textit{Distributionally Robust Fair Optimization} (DRFO), which minimizes the worst-case unfairness over all potential probability distributions of missing sensitive attributes instead of the reconstructed one to account for the impact of the reconstruction errors.
We provide theoretical and empirical evidence to demonstrate that our method can effectively ensure fairness in recommender systems when only limited sensitive attributes are accessible.
% We release our code at: \url{https://github.com/TianhaoShi2001/DRFO}
% We release our code at: \url{https://github.com/TianhaoShi2001/DRFO}

% 直接将重建的敏感属性应用于公平感知的方法是不可行的，这是因为重建的敏感属性无法达到完美，带来对公平的有偏评估和约束。而且，由于敏感属性重建的隐私侵犯风险，有法律规定敏感属性重建需要得到用户许可，因此可能不能通过重建获得完整的敏感属性。
% 为了克服这些挑战，我们DRFO，它结合敏感属性重建和dro技术。使用预测的敏感属性估计一个界，然后优化上界以保证。我们进行理论和实验证明xx.。
% Nevertheless, collecting enough such information is often impractical in real-world scenarios due to privacy protection and the need for registration when collecting information.
% To this end, existing fairness-aware methods typically require enough sensitive attributes and can be effective in enhancing the fairness of recommender systems with these attributes. 
% However, it is often impractical in real-world scenarios due to privacy protection and the difficulty of collecting such information, which limits the effectiveness of existing fairness-aware methods. 
% To address the issue, we propose to reconstruct missing sensitive attributes with historical interaction. 
%  require enough sensitive attributes and can be effective in enhancing the fairness of recommender systems with these attribut
\end{abstract}

% todo
\vspace{-2pt}


\vspace{-2pt}
% click-through rate,


\maketitle
\section{Introduction}
% para1
% A recommender system is designed to provide users with personalized suggestions of items and information that are relevant to their interests and preferences. In various domains of daily life, including online advertising, e-commerce, job searching, insurance, and beyond, recommender systems have become increasingly vital. However, the presence of algorithmic bias and biased data in recommender systems may lead to unfair recommendations, discriminating against certain groups. For instance, studies~\cite{racial-discrimination} have shown that advertising recommendations may be biased, with users whose names are commonly associated with Black individuals being more likely to receive advertisements insinuating a criminal record, while users with names typically associated with White individuals are more likely to receive neutral advertisements. Therefore, it is crucial to maintain fairness in recommender systems.

%%%%%%% revise by chatgpt+zy
Recommender system has become a vital technology in various Web applications, including job searching~\cite{lambrecht2019algorithmic}, online advertising~\cite{wu2021fairness}, and e-commerce~\cite{wang2016your}. By providing personalized information filtering based on user interests, recommender systems have significant social influence~\cite{sml}. 
% However, inherent algorithmic bias and data bias may lead to unfair 
Unfair recommendations that discriminate against specific user groups~\cite{recfair-survey} would negatively impact the Web and society. For instance, previous research~\cite{lambrecht2019algorithmic} on job recommendations found that women were exposed to fewer advertisements for high-paying jobs and career coaching services compared to men, perpetuating gender pay gap in the workforce.
% For example, studies have shown that there is severe discrimination in Facebook's job advertising recommendations, where the target audience for logging industry advertisements are predominantly white males, while for taxi driver advertisements, they are predominantly black. 
% For instance, advertising recommendations may exhibit prejudice, with users whose names are more common among Black individuals more likely to receive ads related to criminal records, while those with names more common among White individuals tend to receive neutral ads~\cite{}. 
Clearly, it is of paramount importance to maintain fairness in recommender systems to ensure their trustworthiness and regulatory compliance~\cite{gdpr}. %应用一个规定保证公平性，不歧视的文章或者法规

% A recommender system designed is to provide users with personalized suggestions of items and information that are relevant to their interests and preferences. 
% recommender systems have become increasingly vital in various domains of daily life, ranging from online advertising, to e-commerce, job searching, insurance, and beyond. 
% However, the presence of algorithmic bias and biased data in recommender systems may discriminate against certain groups, leading to unfair recommendations. 
% For instance, studies~\cite{racial-discrimination} have shown that in advertising recommendations, users with names commonly associated with Black individuals are more inclined to receive advertisements insinuating a criminal record, while users with names typically associated with White individuals are more likely to receive neutral advertisements.
% Therefore, it is crucial to maintain fairness in recommender systems.

% para 2
%% 这里总结句子大概可以突出：消除推荐结果与group之间的相关性？
% 1.主流的研究推荐公平的方法主要修改推荐模型直接生成公平的推荐，或是对推荐模型生成的推荐结果进行微调，根据使用方法的不同，主要可以分为以下几类:
% regularition
% adversarial：
% 基于强化学习的方法:这类方法把推荐过程建模为马尔科夫过程，通过在每个iteration公平约束下优化，或者设计公平奖励的方式，保证推荐系统动态和长期公平。
The existing research to optimize the fairness of recommendations can be further classified into four main categories:
% primarily focuses on fair pre-processing, fair learning, or fair adjustment, which
% fairness criteria can be formulated as regularizers or constraints to guide the process of model optimization
%%% todo revise the paragraph 这一段整体没想清楚，需要重新组织一下。
% 关于三类方法的点评（优缺点），get不到啥作用。
% 介绍三类方法的思路中，第一个、第三个都体现不出来依赖知道sensitive attributes

\begin{itemize}[leftmargin=*]
    % regulartization or constraint method
    \item Data-oriented methods~\cite{ekstrand2018all,rastegarpanah2019fighting}, which adjusts the training data according to sensitive attributes by resampling~\cite{ekstrand2018all} and adding antidote data~\cite{rastegarpanah2019fighting}, etc. 
    % the data manipulation techniques (\eg resampling, adding antidote data), mitigating bias in data and fostering fair treatment.
    % which incorporate sensitive attributes to modify the training data by resampling or adding antidote data, mitigating discrimination and fostering fair treatment.
    \item Adversarial learning-based methods~\cite{wu2021fairness,fairlisa}, which learn fair representations with a min-max game to prevent encoding of sensitive attributes in user embeddings or graph structures.
    
    \item Regularization-based methods~\cite{yao2017beyond,recoindependence,tensorfair,liu2020balancing}, which incorporate fairness metrics such as absolute difference in utility between groups with different sensitive attributes into the objective function~\cite{yao2017beyond,recoindependence,tensorfair} or the reward function~\cite{liu2020balancing}.
    %  Explicit mitigation is achieved by formulating the objective function with fairness
% constraints or regularizers, which are often inspired by the fairness measurements and designed to enforce the
% predictions to be less dependent on the sensitive attributes. The explicit methods are often lexible and easy
% to implement as they only require minimum modiications to objective functions. Some example regularizers
% include co-variance relationships [47], absolute correlation regularization [11], Wasserstein-1 distances [43], etc
    % \item Reinforcement learning-based methods, which model the recommendation problem as a Markov Decision Process and promote fairness by designing fairness-related reward or constraining fairness at each iteration. This paradigm ensures long-term and dynamic fairness in the recommender system, but may also cause poor stability.
    \item Re-ranking methods~\cite{user-orient,wu2021tfrom}, which directly adjusts recommendation results to enhance fairness across user groups with different sensitive attributes. 
    % reorder the recommendation results to achieve a more equitable outcome for 

    % These methods have the advantage of not modifying the recommendation model and being decoupled from the recommendation task.
    % Adversarial learning-based methods, which use a discriminator to prevent the recommendation model from predicting sensitive attributes based on user embeddings to learn fair representations. Despite being well-suited for representation learning, their minimax optimization makes them more challenging to optimize.
\end{itemize}
\vspace{-3pt}
We contend that the current methods are significantly constrained by their reliance on full access to sensitive attributes~\cite{recfair-survey}, a condition that is often difficult to satisfy in real-world scenarios.
Firstly, users have the entitlement, as per regulatory frameworks like the General Data Protection Regulation (GDPR)~\cite{gdpr} of the European Union, to decline the disclosure of their sensitive data at any moment. 
For instance, only 17\% of participants reported willing of sharing their income information in electronic commerce~\cite{privacy-study}.
Moreover, many online retail and advertising recommendations can be accessed without registering or submitting personal information~\cite{wang2016your}.
% It is evident that these methods fail to ensure overall fairness in the absence of some sensitive attributes.
Consequently, a conundrum emerges: how can fair recommendation be obtained when only limited sensitive attributes are available~\cite{recfair-survey}?
To address this issue, a default choice is reconstructing the missing sensitive attributes of users from the available personal information such as historical interactions~\cite{liu2019efficient}.
% and zip codes~\cite{zip-code}.
% , which can be achieved by utilizing personalized information~\cite{liu2019efficient,Yuan_He_Karatzoglou_Zhang_2020,fairtutorial} (\eg historical interactions, zip codes). 
However, reconstruction errors are inevitable due to the inherent complexity and noise of user information in recommender systems.
% Nevertheless, reconstruction errors are inevitable due to two main challenges.
% Firstly, the complexity of users' information in recommender system makes it difficult to achieve perfect reconstruction of sensitive attributes.
Moreover, striving for perfectly reconstructed attributes can potentially violate user privacy, raising ethical concerns and being constrained by legal limitations.
% pursuing perfect reconstructed attributes may infringe on user privacy, which raises ethical concerns and is subject to legal limitations. 
For instance, GDPR~\cite{gdpr} mandates obtaining user consent for modeling user profiles (including sensitive attributes). Consequently, a subset of users may not have their sensitive attributes reconstructed\footnote{This can also be regarded as a special scenario with significant reconstruction errors.}. 
% which can also be regarded as a special scenario with significant reconstruction errors.
%这限制了对一部分用户通过重建获取其敏感属性，也可以看成是重建误差特别大的特殊情况。
The aforementioned challenges underscore the significance of devising alternative approaches that are robust to the reconstruction errors of sensitive attributes.
% These challenges hinder the achievement of fairness through reconstruction methods and highlight the importance of developing alternative approaches that do not rely solely on sensitive attribute reconstruction.

%%%%%%paragraph 5Distributionally Robust Optimization for Fairness with Limited Sensitive Attributes
% In this light, we pursue fair learning methods that are robust to reconstruction errors. 
To this end, we propose a new framework to pursue fair recommendations under reconstructed sensitive attributes with errors called \textit{Distributionally Robust Fair Optimization} (DRFO). 
Instead of optimizing fairness over the reconstructed sensitive attributes, DRFO minimizes the worst-case unfairness over an ambiguity set of all potential distributions of missing sensitive attributes to account for the impact of reconstruction errors.
Furthermore, for users who forbid the reconstruction of their sensitive attributes, DRFO can provide fair recommendations for them by considering a larger ambiguity set of distributions.
% fairness is ensured through a similar approach but with a larger ambiguity set of distributions for sensitive attributes, considering the limited information available.
 We theoretically demonstrate that DRFO ensures recommendation fairness in the face of sensitive attribute reconstruction errors, even in the absence of such reconstruction.
 % evidence to demonstrate that 
 % when only limited sensitive attributes are accessible. 
 Extensive experiments on two datasets verify the effectiveness of our approaches.

The main contributions are summarized as follows:
\begin{itemize}[leftmargin=*]
% \item We introduce a novel problem of recommendation fairness under limited sensitive attributes.
\item We propose a new framework for optimizing the user-side fairness in recommender systems with limited sensitive attributes.
\item We provide theoretical evidence that our method can ensure fair recommendations in the face of sensitive attribute reconstruction errors, even in the absence of such reconstruction.
% conducted a theoretical analysis and
\item We conduct extensive experiments on two datasets, validating the rationality and effectiveness of the proposed framework.
\end{itemize}
\vspace{-3pt}
\section{Preliminaries}
% In this section, we 1.
% 这篇工作研究在用户敏感属性缺失下的公平问题，接下来

% This work studies fairness in recommender systems when only limited sensitive attributes of users' are accessible. 
% In the following, we first show the notations used in this paper, and then 

In this study, we aim to achieve fairness in recommender systems with limited sensitive attributes. We consider the widely used Demographic Parity (DP) as an exemplary fairness criterion to investigate this problem. In the following, we first introduce the notation used in this paper and then provide a brief overview of DP fairness.
%%%%%%%%%%0807
% In this study, we investigate how to achieve fairness in recommender systems with limited sensitive attributes, where the system can only obtain such attributes for a subset of users. 
% We use the classic Matrix Factorization (MF) model as an illustrative example to investigate this new problem and 
% consider the widely used \textit{Demographic Parity} (DP) as an exemplary fairness criterion. In the %subsequent subsections, 
% following, we first introduce the notation used in this paper and then provide a brief overview of DP fairness.
% % , and Distributionally Distribution Optimization (DRO).

\vspace{-2pt}
\subsection{Notation}
%\in \mathcal{U}=\left\{1,\dots,l\right\}
%\in \mathcal{V}=\left\{1,\dots,m\right\}\in \mathcal{S}=\left \{ 1,\dots,s \right \}
% \noindent\paragraph{Notations.}
% \noindent\textbf{Notations.}
% In this work, without clear explanation, we use an uppercase character (\eg $R$) to denote a random variable, the corresponding lowercase character (\eg $r$) to denote its specific value and the corresponding character in calligraphic font (\eg $\mathcal{R}$) to represent the sample space of the variable.
In this study, we use uppercase letters (e.g., $R$) to represent random variables, lowercase letters (e.g., $r$) for specific values of these variables, and calligraphic fonts (e.g., $\mathcal{R}$) to represent the sample space of the variable.
% Let $U$, $V$, $S$, and $R$ denote user, item, sensitive feature, and rating respectively. We denoted the collected interaction data as $\mathcal{D}=\{(u_{i},v_{i},r_{i},s_{i})|i=1,\dots,n\}$,  where $(u_{i},v_{i},r_{i},s_{i})$ denotes the $i$-th sample, $n$ is the size of $\mathcal{D}$, $u_{i}\in \mathcal{U}$, $v_{i}\in \mathcal{V}$, $r_{i}\in \mathcal{R}$ and $s_{i}\in\mathcal{S}$. In this work, we consider the rating is binary, \ie $\mathcal{R}=\{0,1\}$, and assume the sensitive feature is also binary but could be unknown for partial users, \ie $\mathcal{S}=\{0,1, unknown\}$. We could split $\mathcal{D}$ into two parts: $D_{k}$ with known sensitive features and $D_{L}$ with unknown sensitive features, and we have $\mathcal{D}=\mathcal{D}_{k} \cup \mathcal{D}_{L}$. 
% We denote the size of $\mathcal{D}_{k}$ and $\mathcal{D}_{L}$ as $n_{k}$ and $n_{L}$, respectively, which means $n_{k}=|\mathcal{D}_{k}|$ and $n_{L}=|\mathcal{D}_{L}|$.
% Note that a user can only appear to $\mathcal{D}_{k}$ or  $\mathcal{D}_{k}$, not both of them, as the interactions of a user should have the identical sensitive feature. Our goal is to train a model based on $\mathcal{D}$ to achieve fair recommendation for all users.
Let $U$, $V$, $S$, and $R$ denote the user, item, user sensitive attribute, and rating, respectively.
Let $\mathcal{D}$ denote the historical data. Each sample within $\mathcal{D}$ is denoted as $(u,v,s,r)$, where $u\in \mathcal{U}$, $v\in \mathcal{V}$, $r\in \mathcal{R}$, and $s\in\mathcal{S}$.
% Let $\mathcal{D}$ denote the historical data. We denote each sample in $\mathcal{D}$ as $(u,v,s,r)$, and we have $u\in \mathcal{U}$, $v\in \mathcal{V}$, $r\in \mathcal{R}$, and $s\in\mathcal{S}$. 
% The historical data is denoted as $\mathcal{D}=\{(u,v,s,r)|i=1,\dots,n\}$, where $n$ is the size of $\mathcal{D}$, and we have $u\in \mathcal{U}$, $v\in \mathcal{V}$, $r\in \mathcal{R}$, and $s\in\mathcal{S}$. 
In this work, we consider the binary rating, \ie $\mathcal{R}=\{0,1\}$. 
% In this work, we consider rating to be binary, i.e., $\mathcal{R}=\{0,1\}$. 
Additionally, we assume the sensitive feature to be binary but potentially unknown for a subset of users. We split $\mathcal{D}$ into two parts: $\mathcal{D}_{k}$ with \textit{known} sensitive features and $\mathcal{D}_{m}$ with \textit{missing} sensitive features, and we have $\mathcal{D}=\mathcal{D}_{k} \cup \mathcal{D}_{m}$ and $\mathcal{D}_{k}\cap \mathcal{D}_{m}= \varnothing $. For convenience, we denote the samples in $\mathcal{D}$ with the sensitive feature $S=s$ as $\mathcal{D}^{(s)}$, similar for $\mathcal{D}_{k}^{(s)}$ and $\mathcal{D}_{m}^{(s)}$.
\vspace{-2pt}
\subsection{Fairness}
% 推荐系统中，根据利益方与场景不同，公平的定义众多。本文研究在用户敏感属性未知情况下的公平，所以聚焦在用户侧组公平问题中。本文采用经典的dp指标，在推荐中，它被用于衡量模型为不同敏感属性的用户提供的推荐是否独立，提升该指标可以让模型不依赖特定的敏感属性进行推荐，从而满足法律规定或者用户的反歧视要求。例如，在职业推荐中不考虑用户种族。
%例如，不依靠种族，性别进行职业推荐。
% In recommendation systems, fairness can be defined in many ways depending on the stakeholders and context. 
\textit{Demographic Parity}~\cite{dp} is a widely studied fairness criterion in recommendation~\cite{recoindependence,tensorfair}. \sth{DP's definition is the model's rating prediction $\hat{R}$ should be independent of the sensitive attribute $S$.}
A model achieving DP fairness would generate recommendations without relying on $S$, 
 thereby satisfying legal requirements or user demands against discrimination on model output~\cite{recoindependence}.
%ensuring compliance with legal requirements and user expectations by preventing biased model outputs~\cite{recoindependence}.
% To quantify DP, we follow previous work~\cite{tensorfair} with the metric of the absolute difference between mean ratings of different groups (MAD):
We follow previous work~\cite{tensorfair} to quantify DP with the mean absolute difference (MAD) between ratings of different groups:
% To quantify DP, we adopt the metric of the mean absolute difference (MAD) between ratings of different groups, following prior work~\cite{tensorfair}:
% By achieving this fairness, the model can make recommendations without relying on the sensitive attribute $S$~\cite{recoindependence}, thereby satisfying legal requirements or user demands against discrimination. 
% To perform a quantitative assessment of statistical parity, following~\cite{tensorfair}, we use the absolute difference between mean ratings of different groups (MAD):
% According to ~\cite{zafar2019fairness}, we apply the following metrics to perform a quantitative assessment of statistical parity:
\begin{equation}
\small
\vspace{-2pt}
\label{eq: MAD}
    \left | \mathbb{E}\left [\hat{R}|S=0 \right ]-\mathbb{E}\left [\hat{R}|S=1\right ] \right |,
% \vspace{-1pt}
\end{equation}
where $\mathbb{E}\left [\hat{R}|S=0 \right ]$ and $\mathbb{E}\left [\hat{R}|S=1 \right ]$ denote the expectation of prediction $\hat{R}$ over groups with $S=0$ and $S=1$, respectively. 
A diminished MAD level signifies a heightened degree of DP fairness.
% For simplicity, we assume that S is a binary variable, and for non-binary variables, fairness can be achieved by applying multiple constraints simultaneously.
% Then, training a fairness-aware 

\vspace{1pt}
\noindent\textbf{Fair learning}.
% \textit{Fair learning}.
To achieve DP fairness, we could take the regularization-based method~\cite{yao2017beyond}, which directly incorporates the MAD metric into the training objective. Formally,
% Then, considering the convenience and scalability of regularization-based methods, a fairness-aware MF learning objective can be formulated as follows~\cite{yao2017beyond}:
\begin{equation}
\small
\label{eq:mf-fair-reg}
\vspace{-2pt}
    \min_{\theta} \  L \left (\theta \right )  + \lambda 
    \left | 
\mathbb{E}_{\mathcal{D}^{(0)}}\left [\hat{R}\right]- \mathbb{E}_{\mathcal{D}^{(1)}}\left [\hat{R} \right]\right |,
\vspace{-2pt}
%\frac{1}{\left | \mathcal{D}^{(0)} \right |}\sum_{(u,v,s,r)\in \mathcal{D}^{(0)}}\hat{R}\left(U,V \right)
\end{equation}
where $\lambda$ is a hyper-parameter to control the strength of the fairness regularization term,  $\mathbb{E}_{\mathcal{D}^{(s)}}\left [\hat{R}\right]$ is the average predicted rating over $\mathcal{D}^{(s)}$, \ie $\mathbb{E}_{\mathcal{D}^{(s)}}\left [\hat{R}\right]=\frac{1}{|\mathcal{D}^{(s)}|} \sum_{(u,v)\in \mathcal{D}^{(s)}}\hat{r}_{u,v}$
, and $L(\theta)$ is a recommendation loss (\eg binary cross-entropy loss~\cite{ncf}).
\vspace{1pt}
% where $\lambda$ is a hyper-parameter to control the strength of the fairness regularization term, and $\mathbb{E}_{\mathcal{D}^{(s)}}$ ($s=0,1$) is the average of the predicted rating over $\mathcal{D}^{(s)}$, 
% \begin{equation}\small
% \label{eq:emp-mean}
% \vspace{-3pt}
% \mathbb{E}_{\mathcal{D}^{(s)}}\left [\hat{R}\right]=\frac{1}{|\mathcal{D}^{(s)}|} \sum_{(u,v)\in \mathcal{D}^{(s)}}\hat{r}_{u,v}.
% \vspace{-3pt}
% \end{equation}
% \begin{equation}\small
% \label{eq:emp-mean}
% \vspace{-3pt}

% \end{equation}
% | \mathcal{D}^{(s)}  |
% where $\frac{1}{ }$ denotes the size of $\mathcal{D}^{(s)}$.
% denotes the expected value of the predicted 
% preference score, and $\left | \mathcal{D}^{(s)} \right |$ denotes the number of instances of $\mathcal{D}^{(s)}$.
% 考虑到dp也可以表示为xxx，我们可以将正则项等价地转化为约束某一组平均预测得分与全体预测得分的差以保证公平
% This regularization-based optimization problem can be transformed equivalently to an optimization problem with the fairness constraints as follows:
% Give that demographic parity (Equation~\eqref{eq:dp}) can be expressed as
Typically, the regularization can be transformed into a set of constraints that minimize the discrepancy between the average predictions of a specific group and the overall predictions,
% that enforce the fairness by minimizing the difference between the average predicted scores of a specific group and the overall predicted scores as follows.
\begin{equation}
\small
\label{eq:fair-constraint}
\begin{aligned}
\vspace{-2pt}
& \underset{\theta}{\text{min}}
& & L \left (\theta \right ) \\
% = \beta (||\mathbf{P}||^2 + ||\mathbf{Q}||^2) + \frac{1}{\left | \mathcal{D} \right |} \sum_{(u_i,v_i,r_i)  \in \mathcal{D}} \left [-r_{i}\log(\hat{r}_{u_i,v_i}) - (1-r_{i})\log(1-\hat{r}_{u_i,v_i}) \right ]  \\
& \text{s.t.}
& & 
\mathbb{E}_{\mathcal{D}^{(s)}}\left [\hat{R}\right]- \mathbb{E}_{\mathcal{D}}\left [\hat{R}\right] = 0, \quad  s=0,1,
\vspace{-4pt}
\end{aligned}
\end{equation}
where the constraint ensures the expected predicted rating in $\mathcal{D}^{(s)}$ is equal to the expected predicted rating in the entire dataset $\mathcal{D}$. 
Preserving any single constraint in Equation~\eqref{eq:fair-constraint} is sufficient to promote fairness under the binary-sensitive attribute scenario while preserving multiple constraints is intended for non-binary cases.

\vspace{-3pt}
\section{METHODOLOGY}
To build a fair recommender system that addresses the chanllenge of missing sensitive attributes among some users, a seemingly workable solution is to directly apply Fair Learning with Reconstructed Sensitive Attributes (FLrSA). In this section, we outline this solution and highlight its limitations.
Subsequently, we introduce the proposed Distributionally Robust Fairness Optimization to overcome the impact of reconstruction errors. Finally, we discuss the extension of DRFO for situations where certain users are reluctant to have their sensitive attributes reconstructed.
\vspace{-3pt}
\subsection{FLrSA}
\label{sec:Intuitive}

% 考虑到现有的公平方法在缺失敏感属性的情况下效果不好，我们
% 考虑到推荐系统中用户历史交互信息可以反映用户本身的特点，可以利用历史交互信息预测敏感属性（引用）。
% To address the challenge of guaranteeing fairness in recommendation systems when only limited sensitive user attributes are accessible, 
%%%%%%%%%%%%
% Considering that existing fairness methods for recommendation may not perform well in the absence of abundant sensitive attribute information, an intuitive approach to ensuring recommendation fairness under limited sensitive attribute scenarios is to reconstruct the missing sensitive attributes using other data that contain sensitive attribute information. 
% As user historical interaction information in recommendation systems can reflect users' characteristics, it is possible to reconstruct sensitive attributes using user's historical interaction information~\cite{xxx}. In practice, one can use logistic regression, SVM, or neural networks to reconstruct the missing sensitive attributes based on the users' interaction vectors or the user embeddings as input. 
%%%%%%%%%%%
% Existing fairness methods in recommendation require complete information on sensitive attributes to achieve fair recommendations, which can become ineffective when some users' sensitive attributes are missing. 
To achieve fair recommendations in the presence of missing sensitive attributes, a seemingly workable solution involves reconstructing the missing attributes and subsequently applying fair learning methods based on the reconstructed sensitive attributes. Specifically, we can follow the two steps below:

\textit{Step 1: Sensitive attribute reconstruction}. 
% Previous work has shown that user historical interactions could reflect user users' characteristics~\cite{xxxx}, we thus could reconstruct the sensitive attribute of users with the historical interactions. Specifically, we first fit the data knowing sensitive attributes (\ie $\mathcal{D}_{k}$) to train a classifier, which takes user historical interactions as input and outputs the prediction of the target sensitive attribute for the user. Here, the classifier could be any classification model, such as SVM~\cite{xx} and neural networks~\cite{xx}. 
% After finishing the training, we could take the trained classifier to reconstruct the sensitive attribute for users in $\mathcal{D}_{m}$, and we take $\hat{S}\in \{0,1\}$ to represent the reconstructed sensitive attribute. For each sample $(u_i,v_i,r_i,s_i) \in \mathcal{D}_{m}$, we replace $s_i$ with the reconstructed one $\hat{s}_{i}$, forming $\hat{\mathcal{D}}_{m}=\{(u_{i},i_{i},r_{i},\hat{s}_{i})|i=1,\dots,n_{m}\}$. 
%%%%%%%%%%%%
Previous research has shown that user sensitive attributes can be reconstructed using available user information, such as historical interactions~\cite{weinsberg2012blurme,wang2016your} and zip codes~\cite{zip-code}. Therefore, we can train a classifier over $\mathcal{D}_k$ to predict missing user sensitive attributes in $\mathcal{D}_{m}$. Let $\hat{S}\in \{0,1\}$ denote the reconstructed sensitive attribute. 
\sth{Subsequently, we can incorporate the reconstructed sensitive attributes into $\mathcal{D}_m$ and generate a new dataset $\hat{\mathcal{D}}_{m}$.}

\textit{Step 2: Fair learning}. Next, we perform fair learning over $\mathcal{D}_k$ and $\hat{\mathcal{D}}_m$ based on Equation~\eqref{eq:fair-constraint}, which is reformulated as:
\begin{equation}
\small
\label{eq:naive-fair-constraint-with-reconstructed}
\begin{aligned}
& \qquad \qquad \qquad \qquad \underset{\theta}{\text{min}}
\  L \left (\theta \right ) 
%= \beta (||\mathbf{P}||^2 + ||\mathbf{Q}||^2) + \frac{1}{\left | \mathcal{D} \right |} \sum_{(u_i,v_i,r_i)  \in \mathcal{D}} \left [-r_{i}\log(\hat{r}_{u_i,v_i}) - (1-r_{i})\log(1-\hat{r}_{u_i,v_i}) \right ]  \\
\\
& \text{s.t.}
\quad \eta_{k}^{(s)}\mathbb{E}_{\mathcal{D}_{k}^{(s)}}\left [\hat{R}\right] + \eta_{m}^{(s)}\mathbb{E}_{\hat{\mathcal{D}}_{m}^{(s)}}\left [\hat{R}\right] - \mathbb{E}_{\mathcal{D}}\left [\hat{R}\right] = 0, \  s=0,1, 
\end{aligned}
\vspace{-2pt}
\end{equation}
where $\hat{\mathcal{D}}^{(s)}_{m}$ is a subset of $\mathcal{\hat{D}}_{m}$ with the reconstructed attribute $\hat{S}=s$, $  \mathbb{E}_{\mathcal{D}_{k}^{(s)}}\left [\hat{R}\right]$ is the average predicted rating over $\mathcal{D}_{k}^{(s)}$, and:
\begin{equation}
\small
\vspace{-2pt}
\label{eq:intuitive-solution}
    \eta_{k}^{(s)}=\frac{|\mathcal{D}^{(s)}_{k}|}{|\mathcal{D}_{k}^{(s)}|
      + | \hat{\mathcal{D}}_{m}^{(s)}|}, \,  \eta_{m}^{(s)}=\frac{|\hat{\mathcal{D}}_{m}^{(s)}|}{|\mathcal{D}_{k}^{(s)}| +|\hat{\mathcal{D}}_{m}^{(s)}|}.
\end{equation}
% \begin{equation}\small
% \label{eq:intuitive-solution}
%     \eta_{k}^{(s)}=\frac{\left| \mathcal{D}^{(s)}_{k}\right|}{\left| \mathcal{D}_{k}^{(s)}\right| + \left| \hat{\mathcal{D}}_{m}^{(s)}\right|}, \,and \quad \eta_{m}^{(s)}=\frac{\left| \mathcal{\hat{D}}_{m}^{(s)}\right|}{\left| \mathcal{D}_{k}^{(s)}\right| +\left| \hat{\mathcal{D}}_{m}^{(s)}\right|}.
% \end{equation}

% \noindent{Limitations.} As shown in Equation~\eqref{eq:intuitive-solution}, this method directly takes the reconstructed sensitive attributes as the exact sensitive attributes. However, it is near impossible to build a completely accurate classifier, considering the information gap between model input (user historical interactions) and sensitive attributes, which means reconstruction errors always exist. Therefore, the method cannot make sure the achievement of fair recommendation for users. 
\noindent\textbf{Limitations.} This method relies on the accuracy of sensitive attribute reconstruction.
% However, it is difficult to build a completely accurate classifier due to the information gap between model input (user historical interactions) and user sensitive attributes~\cite{xxx}. 
% reconstruction errors are inevitable due to the inherent complexity and noise of user information in recommender systems.
However, achieving an entirely accurate classifier poses challenges due to the inherent complexity and noise of user information in recommender systems.
This leads to reconstruction errors, which could compromise the fairness of recommendations for some users.
Furthermore, this method relies on the permission of sensitive attribute reconstruction from users, which may not be achievable for all users due to legal restrictions.
\subsection{DRFO}
\label{sec:DRFO}
\sth{Reconstruction errors significantly constrain the vanilla FLrSA, as they introduce a discrepancy between the reconstructed and the unknown true distribution concerning sensitive attributes.} Consequently, relying solely on the reconstructed distribution may compromise fairness performance. 
Nevertheless, the unknown true distribution lies within the proximity of the reconstructed distribution. By ensuring fairness in the vicinity of the reconstructed distribution, the model could achieve robust fairness for the unknown true distribution~\cite{wang2020robust}.
% Nonetheless, we hypothesize that the unknown true distribution lies within an ambiguity set of distributions centered around the reconstructed distribution.
% By ensuring fairness across the entire set of distributions, the model can establish robust fairness for the unknown true distribution.
This inspires the development of DRFO, a novel approach to fairness with limited sensitive attributes.
DRFO has two main parts: 1) building the ambiguity set which encompasses the unknown true distribution based on the reconstructed sensitive attributes, and 2) ensuring fairness within the entire ambiguity set using DRO. Figure~\ref{fig:method} provides an overview of DRFO. \sth{For our discussion convenience, we assume that all users grant permission for the reconstruction of sensitive attributes in this subsection. The scenario where some users do not permit reconstruction due to privacy concerns will be discussed in Section~\ref{sec:discussion}.} 

\begin{figure}[t]
    \centering
    \includegraphics[width=0.45\textwidth]{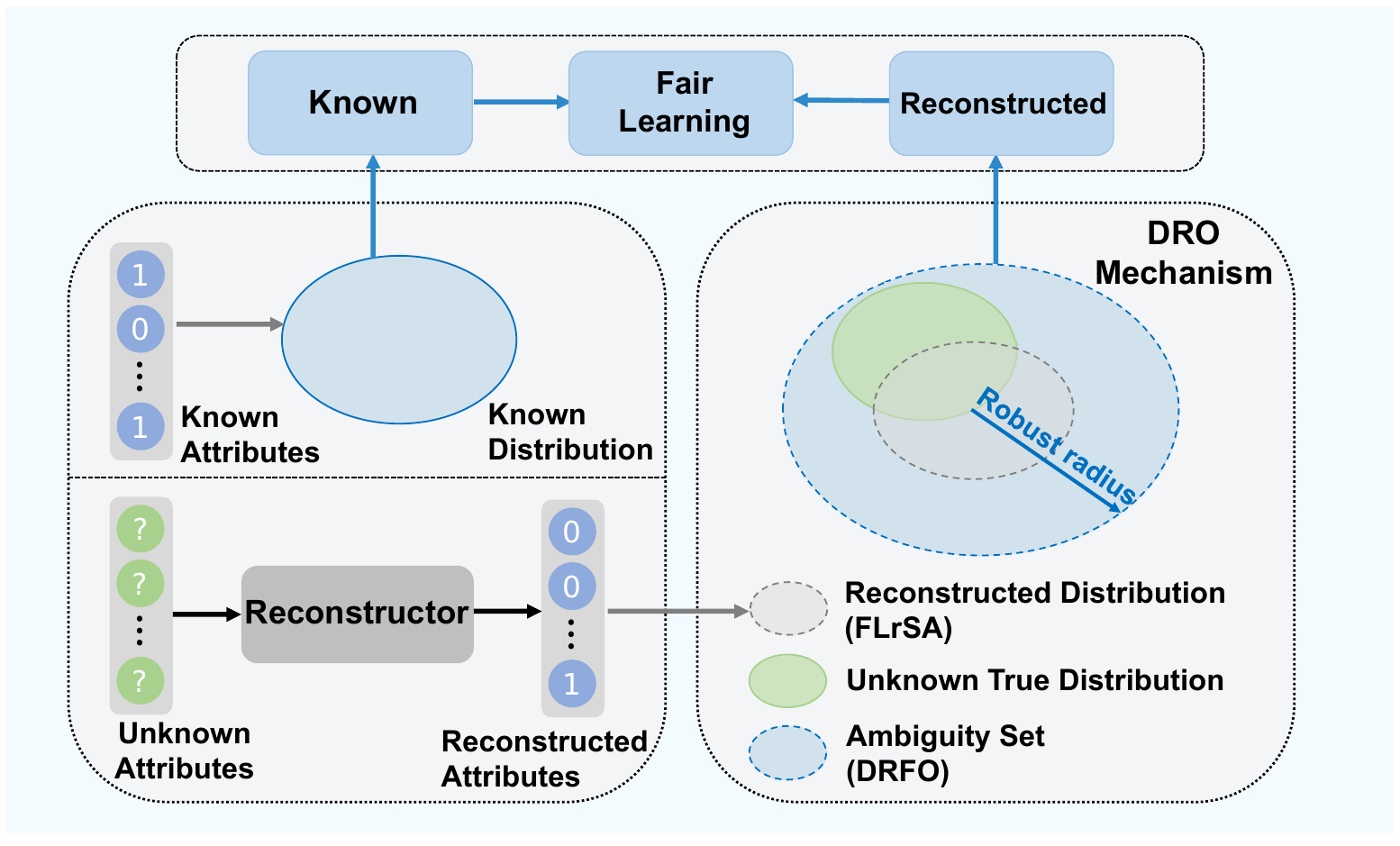}
    % \vspace{-4pt}
    \vspace{-5pt}
    \caption{Illustration of FLrSA and DRFO for providing fair recommendations with limited sensitive attributes. After the reconstruction of unknown sensitive attributes, the FLrSA directly applies fair learning with the reconstructed distribution. Conversely, DRFO builds an ambiguity set that encompasses the unknown true distribution and guarantees fairness across the entire ambiguity set.}
    \vspace{-12pt}
    \label{fig:method}
\end{figure}

% \vspace{pt}
\noindent \textbf{Building ambiguity set}. 
An ambiguity set is a set of distributions centered around the reconstructed distribution. We denote the ambiguity set as $\mathbb{B}(\rho_{s}; \hat{Q}^{(s)}) \text{=} \{\tilde{Q}^{(s)}|dist(\tilde{Q}^{(s)}, \hat{Q}^{(s)})\leq \rho_{s}\}$, where $\hat{Q}^{(s)}$ denotes the reconstructed distribution, $dist(\cdot)$ denotes a distance metric, and $\rho_{s}$ denotes robust radius. By setting an appropriate value $\rho_{s}$, we can ensure that this set encompasses the unknown true distribution $Q^{(s)}$.
In the context of facilitating fair learning, as expressed in Equation~\eqref{eq:naive-fair-constraint-with-reconstructed} involving $\mathbb{E}_{\hat{\mathcal{D}}_{m}^{(s)}}[\hat{R}(U,V)]$, our method focuses on building an ambiguity set of joint distributions $(U,V)$ conditioned on a specific sensitive attribute value $s$. Specifically, we assume that the distribution of $(U,V)$ conditioned on $S\text{=}s$ follows $Q^{(s)}$. And we represent the joint distribution of $(U,V)$ conditioned on the reconstructed sensitive attribute $\hat{S}\text{=}s$ as $\hat{Q}^{(s)}$.

\vspace{3pt}
\noindent \textbf{Robust fair learning.} 
% \sth{We incorporate Equation~\eqref{eq:naive-fair-constraint-with-reconstructed} with a DRO-based method to optimize the recommender model, ensuring the satisfaction of fairness constraints within the entire set $\mathbb{B}(\rho_{s}; \hat{Q}^{(s)})$. Formally,}
Robust learning solves an optimization problem that satisfies the fairness constraints in $\mathbb{B}(\rho_{s}; \hat{Q}^{(s)})$ as:
 % We next use a DRO-based fair learning method to optimize the recommender model which satisfies that the fairness constraints within $\mathbb{B}(\rho_{s}; \hat{Q}^{(s)})$, which could be formulated as the following optimization problem, incorporating Equation~\eqref{eq:naive-fair-constraint-with-reconstructed}:
% We next optimize the recommender model by making sure that the fairness constraints are always satisfied within  $\mathbb{B}(\rho_{s}; \hat{Q}^{(s)})$, which could be formulated as the following optimization problem, incorporating Equation~\eqref{eq:naive-fair-constraint-with-reconstructed} and leveraging DRO as expressed in Equation~\eqref{eq:fair-constraint-dro}:
\begin{equation}
\small
\label{eq:constraint-fair-in-dist}
\begin{aligned}
& \qquad \qquad \qquad \qquad \underset{\theta}{\text{min}}
\ L \left (\theta \right )\\
& \text{s.t.}
\quad \eta_{k}^{(s)}\mathbb{E}_{\mathcal{D}_{k}^{(s)}}\left [\hat{R}\right] + \eta_{m}^{(s)}\mathbb{E}_{(U,V)\sim {\tilde{Q}}^{(s)}}\left [\hat{R}\right] - \mathbb{E}_{\mathcal{D}}\left [\hat{R}\right] = 0, 
\\
& \qquad \qquad \qquad \forall \tilde{Q}^{(s)} \in \mathbb{B}(\rho_{s}; \hat{Q}^{(s)})\,, s=0, 1 ,
% \\ 
%  & & & \qquad \qquad \forall \tilde{P}^{(s)}_{m}:  dist({\tilde{P}}^{(s)}_{m}, {\hat{Q}}^{(s)}_{m}) \leq \rho_{s}, \  s=0,1 .
\end{aligned}
\end{equation}
where $\mathbb{E}_{(U,V) \sim {\tilde{Q}}^{(s)}}\left [\hat{R} \right]$ denotes the expectation of $\hat{R}$ under the distribution $\tilde{Q}^{(s)}$.
% Here we adjust the fairness constraint si on the reconstructed distribution to satisfying fairness constraints across the entire ambiguity set.
% Here, 我们将等式x中，只基于reconstructed distribution上约束公平，改成了在整个模糊集上都保证公平
Here, we transform the fairness constraints in Equation~\eqref{eq:naive-fair-constraint-with-reconstructed} from solely relying on the reconstructed distribution to guaranteeing fairness across the entire ambiguity set\footnote{\zy{We do not directly solve the optimization problem under these equality constraints in Equation~\eqref{eq:constraint-fair-in-dist}. Instead, we convert it into an optimization of the worst-case scenario problem, as expressed in Equation~\eqref{eq:emp-lag-form}}.}.
In this way, as long as the ambiguity set encompasses the unknown true distribution, we can achieve robust fairness.
% In this way, we could achieve robust fairness, as the set $\mathbb{B}(\rho_{s}; \hat{Q}^{(s)})$ has encompassed the unknown true distribution.

%In DRFO, there are two main challenges: 
Apparently, the key of DRFO lies in
1) building an appropriate ambiguity set which encompasses the unknown true distribution, and 2) solving the optimization problem in Equation~\eqref{eq:constraint-fair-in-dist}. Next, %we provide a detailed explanation of how to address each challenge.
we elaborate how to achieve them.

\subsubsection{Building Ambiguity Set}
%%%%%%%%%%%%%%%%%zyang
 % We next consider establishing an appropriate ambiguity set $\mathbb{B}(\rho_{s}; \hat{Q}^{(s)})$ that encompasses the unknown true distribution $Q^{(s)}$ associated with the unknown true sensitive attributes. 
 % According to the definition of the ambiguity set, the key is to estimate the distance between $Q^{(s)}$ and $\hat{Q}^{(s)}$ (associated with the reconstructed sensitive attributes).  However, since the true sensitive attributes are unknown, directly estimating the exact distance is impossible. Instead, we estimate an upper bound on the distance between them. Similar to previous work, using the Total Variation (TV) distance as the distance metric, we estimate an upper bound based on the error rate of the reconstructed sensitive attributes. This is described in the following theorem. 
%%%%%%%%%%%%%shith0807
We now consider establishing an appropriate ambiguity set $\mathbb{B}(\rho_{s}; \hat{Q}^{(s)})$ that encompasses the unknown true distribution $Q^{(s)}$.
\sth{However, direct estimation of the exact distance between $Q^{(s)}$ and the reconstructed distribution $\hat{Q}^{(s)}$ is impossible as true sensitive attributes are unavailable.} Instead, we estimate an upper bound on the distance between them. Similar to previous work~\cite{wang2020robust}, using the Total Variation (TV) distance~\cite{TV-distance} ($TV \in [0,1]$, Appendix~\ref{sec:appendix}) as the distance metric, we could estimate an upper bound based on the error rate of sensitive attribute reconstruction. This is described in the following theorem.

\setcounter{theorem}{0}
\begin{theorem}
\vspace{-2pt}
\label{theo:tv-distance}
% Assuming $P(S=s)=P(\hat{S}=s)$ for a given $s\in\{0,1\}$, the TV distance between $P_{m}^{(s)}$ and $\hat{Q}{m}^{(s)}$ is upper-bounded by the probability of incorrectly reconstructing the sensitive attributes, \ie $TV(P{m}^{(s)},\hat{Q}{m}^{(s)})\leq P(S\neq\hat{S}|S=s)$. 
% Assuming the reconstructed and true sensitive attributes have the identical distribution, 
Assuming that the reconstructed sensitive attributes $\hat{S}$ have the same prior distribution as the true sensitive attributes $S$,
\ie $P(\hat{S})=P(S)$, 
%$P(S=s)=P(\hat{S}=s)$ for a given $s\in\{0,1\}$, 
the TV distance between $Q^{(s)}$ and $\hat{Q}^{(s)}$ is upper-bounded by the probability of incorrectly reconstructing the sensitive attributes, \ie $TV(Q^{(s)},\hat{Q}^{(s)})\leq P(S\neq\hat{S}|S=s)$. 

% 

% Assuming $P(S=s)=P(\hat{S}=s)$ for a given $s \in \left\{ 0,1 \right \}$, the TV distance between $P_{m}^{(s)}$ and $\hat{Q}^{(s)}$ is bounded by the probability of how likely the sensitive attributes incorrectly reconstructed, i.e.,  $TV(P_{m}^{(s)},\hat{Q}^{(s)})\leq P(S\ne \hat{S}|S=s)$.
% Suppose $P(S=s)=P(\hat{S}=s)$ for a given $s \in \left\{ 0,1 \right \}$. Then $TV(p^{(s)},\hat{Q}^{(s)})\leq P(S\ne \hat{S}|S=s)$
\vspace{-2pt}
\end{theorem}
The proof is provided in Appendix~\ref{sec:appendix}. 
This theorem suggests that, assuming the distribution of the reconstructed sensitive attributes $\hat{S}$ is identical to the true distribution of $S$\footnote{\sth{
If the assumption is violated, a more lenient TV distance upper bound estimation is needed (further discussed in Appendix~\ref{sec:appendix}). In our experiments, as the observed commendable accuracy in reconstructing sensitive attributes leads to a modest disparity between $P(S)$ and $P(\hat{S})$, we approximate the assumption holds true here. 
% allowing us to approximate the assumption holds true.
% In our experiments, we found accurate reconstruction of sensitive attributes with minimal difference between P(S) and P(S), affirming the validity of the assumption.
% Given the high accuracy observed in reconstructing sensitive attributes during our experiments, it is reasonable to consider that $P(\hat{S})$ approximates $P(S)$, justifying the assumption's validity. 
% Hence, assuming an approximate validity of this assumption is justified.
% could approximate $P(S)$ in our experiments, it is reasonable to assume an approximate validity of this assumption.
% 考虑到实验中重建的敏感属性准确率比较高，可以近似还原出P(S)
% Violation of the assumption necessitates a wider Total Variation (TV) upper bound estimate, especially when $P(S)$ and $P(\hat{S})$ diverge. Further details are provided in the appendix. In practice, the assumption is approximated as valid due to the high accuracy in reconstructing sensitive attributes.
}}, we can create the ambiguity set $\mathbb{B}(\rho_{s}; \hat{Q}^{(s)})$ as follows:
% This theorem suggests that, under the assumption that the distribution of the reconstructed and true sensitive attributes are identical, we can create the ambiguity set $\mathbb{B}(\rho_{s}; \hat{Q}^{(s)})$ as follows:
% under the assumption that the distribution of reconstructed $\hat{S}$ is identical to the true prior distribution of $S$, we could build the ambiguity set $\mathbb{B}(\rho_{s}; \hat{Q}^{(s)})$ as follows:
% Theorem~\ref{theo:tv-distance} demonstrates that --- under the assumption that the distribution of reconstructed $\hat{S}$ is identical to the true prior distribution of $S$, we could build the ambiguity set $\mathbb{B}(\rho_{s}; \hat{Q}^{(s)})$ as:
\begin{equation}
\small
    \label{eq: range}
    \mathbb{B}\left(
        \rho_{s}; \hat{Q}^{(s)}
    \right) = \left\{
        \tilde{Q}^{(s)}| TV(\tilde{Q}^{(s)}, \hat{Q}^{(s)}) \leq \rho_{s}
    \right\}, \rho_{s}=P(S\ne \hat{S}|S=s),
\end{equation}
% todo revise
where  $\rho_{s}$ can be approximately estimated using the validation set\footnote{\zy{
% When assessing fairness without sensitive attributes, we adopt the idea of leveraging prior assumptions on sensitive attribute distributions from an auxiliary dataset~\cite{kallus2022assessing}. 
In practice, by assessing the difference between the validation and test sets and incorporating it with the error rate of the validation set, we can set an upper bound on the reconstruction of sensitive attribute errors, thus enabling estimation of $\rho_{s}$.}}, \textcolor{black}{following existing works~\cite{wang2020robust}}. 
This ambiguity set encompasses the true distribution ${Q}^{(s)}$, meaning that we successfully build the target ambiguity set with $\hat{Q}^{(s)}$ and $\rho_{s}$.

% 当假设不满足时，P(S)与P(S)相差的越多，则对TV上界的估计需要越宽松，这在appendix中会被讨论。实践中由于重建敏感属性的准确率比较高，这里近似认为假设成立。
% Additionally, to ensure the assumption of $P(S=s)=P(\hat{S}=s)$ approximately holds, we can add the corresponding constraints on the classifiers.
% Next, we explain how to solve the optimization problem in Equation~(\ref{eq:constraint-fair-in-dist}) with the introduced ambiguity set.

% where $P(S\ne\hat{S}|S=s)$ can be estimated using the validation set of the classifier.
% Obviously, this distribution range covers the true distribution $P_{m}^{(s)}$. Besides, regarding the assumption of $P(S=s)=P(\hat{S}=s)$, it could be approximatively kept by adding the constraints on the classifiers.
% We next explain how to solve the optimization problem in Equation~(\ref{eq:constraint-fair-in-dist}) under this distribution range.

% Theorem~\ref{theo:tv-distance} demonstrates that under the assumption that the reconstructed prior distribution of $\hat{S}$ is identical to the true prior distribution of $S$, we can estimate an upper bound on the distance between probabilities by using the prediction error rate of the classifier for a certain category. 

% In the following section, we will explain how to use this distance to solve the optimization problem in Equation~(\ref{eq:constraint-fair-in-dist}).

% In this section, we will provide a theoretical analysis and show how to estimate the distance bound.

% \subsubsection{Learning Algorithm of DRFO}
\vspace{-2pt}
\subsubsection{Robust Fair Learning}

We next consider solving the optimization problem in Equation~\eqref{eq:constraint-fair-in-dist}.
\zy{Following~\cite{hu2018does,oren2019distributionally}, we convert it into an empirical form (\ie representing it using the empirical distribution), enabling us to optimize it in a data-driven manner. 
Meanwhile, to tackle the challenges posed by the complexity of solving the constrained optimization problem, we transform it into a solvable Lagrangian problem with the algorithm proposed in~\cite{gda}.}
% Meanwhile, considering that the constrained optimization problem is hard to optimize, we convert it into a solvable Lagrangian problem with the algorithm proposed in~\cite{gda}.}
%  with the algorithm proposed in~\cite{gda}, we convert the constrained optimization problem into a solvable Lagrangian problem.} 
% We need to convert it into an empirical form, \ie representing it using the empirical distribution, and dispel the optimization constraints using the Lagrangian trick. In the following, we elaborate on these two steps.
% and then present our detailed learning algorithm.
% to solve the problem.

% After obtaining the ambiguity set as Equation~\eqref{eq: range}, we can solve the optimization problem in Equation~\eqref{eq:constraint-fair-in-dist} using the empirical Lagrangian formulation. To facilitate the understanding of the process, we will first convert the problem into empirical form, and then into empirical Lagrangian form. Finally, we will provide the detailed learning algorithm for the empirical Lagrangian form.

% \vspace{0.5pt}
\vspace{3pt}
\noindent \textbf{Empirical form.}
% To solve the problem using learning methods, we need to convert all expectations of $\hat{R}$ regarding $\hat{Q}{m}^{(s)}$ into the empirical mean of $\hat{R}$ over $\mathcal{D}_{m}$, and convert $\tilde{P}_{m}^{(s)}$ and $\hat{Q}^{(s)}$ into empirical distributions. Specifically, we can represent the empirical distribution as a vector with $n_{m}$ entries, i.e., $\hat{Q}{m}^{(s)}=[\hat{Q}{L,1}^{(s)}, \dots, \hat{Q}{L,i}^{(s)},\dots, \hat{Q}{L,n_L}^{(s)}]$, where $\hat{Q}{L,i}^{(s)}$ corresponds to the $i$-th sample $(u_i,v_i,r_i,s_i)$ in $\mathcal{D}{m}$, and $\hat{Q}{L,i}^{(s)}=\frac{1}{|\mathcal{D}{m}|}$ if $(u_i,v_i,r_i,s_i)\in \mathcal{D}{m}^{(s)}$ else $0$. Similarly, $\tilde{P}{m}^{(s)}=[\tilde{p}{L,1}^{(s)}, \dots, \tilde{p}{L,i}^{(s)},\dots, \tilde{p}_{L,n_L}^{(s)}]$, and then we have:
% To solve the problem with learning methods, we need to convert all expectations of $\hat{R}$ regarding $\hat{P}^{(s)}$ into the empirical mean of $\hat{R}$ over $\hat{\mathcal{D}}_{m}$, and convert $\tilde{P}_{m}^{(s)}$ and $\hat{P}^{(s)}$ into empirical distribution. 
% The key is to use $\hat{\mathcal{D}}_m$ to estimate $\mathbb{E}_{(U,V)\sim {\tilde{Q}^{(s)}}}\left [\hat{R}\right]$, denoted the empirical estimation as $\mathbb{E}_{ \tilde{Q}^{(s)},\hat{\mathcal{D}}_{m}} \left [\hat{R}\right]$. 
% Then, the empirical form of optimization problem~\eqref{eq:constraint-fair-in-dist} can be formulated as follows: 
\sth{To solve Equation~\eqref{eq:constraint-fair-in-dist} with learning methods,} the key step involves using $\hat{\mathcal{D}}_m$ to estimate $\mathbb{E}_{(U,V)\sim {\tilde{Q}^{(s)}}}\left [\hat{R}\right]$. Denote the empirical estimation as $\mathbb{E}_{ \tilde{Q}^{(s)},\hat{\mathcal{D}}_{m}} \left [\hat{R}\right]$, we can obtain the empirical form of the optimization problem as follows: 
\begin{equation}
\small
% \vspace{-2pt}
\label{eq:fair-emp-constraint}
\begin{aligned}
&\qquad\qquad\qquad\qquad\underset{\theta}{\text{min}} \ L \left (\theta \right ) \ \\
\text{s.t. } &\eta_{k}^{(s)}\mathbb{E}_{\mathcal{D}_{k}^{(s)}}\left [\hat{R}\right] + \eta_{m}^{(s)}\mathbb{E}_{\tilde{Q}^{(s)},\hat{\mathcal{D}}_m}\left[\hat{R}\right] - \mathbb{E}_{\mathcal{D}}\left [\hat{R} \right] = 0, \\ 
& \qquad\qquad\qquad \forall \tilde{Q}^{(s)}\in \mathbb{B}({\rho_{s}};\hat{Q}^{(s)}) , \ s=0,1,
\end{aligned}
% \vspace{-2pt}
\end{equation}
where  $\hat{Q}^{(s)}$, $\tilde{Q}^{(s)}$  are also converted into empirical distributions~\cite{dekking2005modern}. $\hat{Q}^{(s)}=\{\hat{q}_{u,v}^{(s)}|(u,v,\hat{s}) \in \hat{\mathcal{D}}_{m}\}$,
% where $\hat{q}_{u,v}^{(s)}$ denotes the probability weight for the sample $(u,v)\in \hat{\mathcal{D}}_m$ and $\hat{q}_{u,v}^{(s)}=\frac{1}{|\hat{\mathcal{D}}_{m}^{(s)}|}$ if the reconstructed $\hat{S}=s$ else 0; similarly, the empirical $\tilde{Q}^{(s)}$ can be represented as $\tilde{Q}^{(s)}=\{\tilde{q}_{u,v}^{(s)}|(u,v) \in \hat{\mathcal{D}}_{m}\}$. 
where $\hat{q}_{u,v}^{(s)}$ denotes the probability weight for the sample $(u,v,\hat{s})\in \hat{\mathcal{D}}_m$ and $\hat{q}_{u,v}^{(s)}=1/|\hat{\mathcal{D}}_{m}^{(s)}|$ if $\hat{s}=s$ else $\hat{q}_{u,v}^{(s)}=0$; similarly, we have $\tilde{Q}^{(s)}=\{\tilde{q}_{u,v}^{(s)}|(u,v,\hat{s}) \in \hat{\mathcal{D}}_{m}\}$; $\mathbb{E}_{\tilde{Q}^{(s)},\hat{\mathcal{D}}_{m}}\left[\hat{R}\right]$ represents the \textit{empirical} expectation of $\hat{R}$ in $\hat{\mathcal{D}}_{m}$ under the distribution $\tilde{Q}^{(s)}$, and $\mathbb{B}\left(
        \rho_{s};\hat{Q}^{(s)}
    \right)$ denotes the empirical form of ambiguity set defined in Equation~\eqref{eq: range}, formally, we have:
    
\begin{itemize} [leftmargin=*]
\item The \textit{empirical} expectation of $\hat{R}$, \ie $\mathbb{E}_{\tilde{Q}^{(s)},\hat{\mathcal{D}}_{m}}\left[\hat{R}\right]$:
\begin{equation}
\small
\label{eq:empirical constraint}\mathbb{E}_{\tilde{Q}^{(s)},\hat{\mathcal{D}}_{m}}\left[\hat{R}\right]=\sum_{(u,v)\in \hat{\mathcal{D}}_{m}} \tilde{q}_{u,v}^{(s)} \cdot\hat{r}_{u,v};
\end{equation}

\item The \textit{empirical} form of ambiguity set $\mathbb{B}\left(
        \rho_{s};\hat{Q}^{(s)}
    \right)$:
    \begin{equation}
    \small
    \begin{aligned}
    \mathbb{B}\left(
        \rho_{s};\hat{Q}^{(s)}
    \right) = \Big\{
        \tilde{Q}^{(s)} \in \mathbb{R}^{|\hat{\mathcal{D}}_{m}|}: \frac{1}{2}\sum_{(u,v)\in \hat{\mathcal{D}}_{m}}\left|\tilde{q}^{(s)}_{u,v} - \hat{q}^{(s)}_{u,v}\right| \leq \rho_{s}, \\
     \sum_{(u,v)\in \hat{\mathcal{D}}_{m}} 
        \tilde{q}^{(s)}_{u,v}=1, \,\tilde{q}_{u,v}^{(s)} \ge 0
    \Big\},
    \end{aligned}
    % \vspace{2pt}
\end{equation}
where $\frac{1}{2}\sum_{(u,v)\in \hat{\mathcal{D}}_{m}}\left|\tilde{q}^{(s)}_{u,v} - \hat{q}^{(s)}_{u,v}\right| \leq \rho_{s}$  represents the empirical implementation of the constraint $TV({\tilde{Q}}^{(s)}$
$, {\hat{Q}}^{(s)}) \leq \rho_{s}$ (Equation~\eqref{eq:tv-distance-probability} in Appendix~\ref{sec:appendix}), 
$\sum_{(u,v)\in \hat{\mathcal{D}}_{m}} 
        \tilde{q}^{(s)}_{u,v}=1$ and $\tilde{q}_{u,v}^{(s)} \ge 0$ are used to ensure that 
%the two terms below it ensure that 
the empirical distribution $\tilde{Q}^{(s)}$ represents a valid probability distribution. 

\end{itemize}

\begin{algorithm}[t]
%\Setvlineskip{-5pt} 
	\caption{DRFO}
 \label{alg:DRFO}
 % \begin{algorithmic}[1]
 %\SetInd{0.5em}{0em} 
	\LinesNumbered
 %\SetAlgoInsideSkip{-10pt}
 %\vspace{-5pt}
	\KwIn{Dataset with known sensitive attributes $\mathcal{D}_{k}$, dataset with unkown sensitive attributes $\mathcal{D}_{m}$. Hyper-parameters $\lambda_{s}$, and learning rate $\alpha_{\theta}$ for $\theta$, learning rate $\alpha_{q}$ for $\tilde{Q}^{(s)}$ ($s=0,1$). }
 Random split $\mathcal{D}_{k}$ into training and validation sets, and train a sensitive attribute classifier\;

%  Reconstruct sensitive attributes $\hat{S}$ for $\mathcal{D}_{m}$, getting $\hat{\mathcal{D}}_{m}$, and  $\hat{Q}^{(s)} \in \mathbb{R}^{|\hat{\mathcal{D}}_{m}|}$ with 
% $\hat{q}^{(s)}_{u,v} = \frac{1}{|\hat{\mathcal{D}}_{m}^{(s)}|}$ if $\hat{s}=s$ for each sample $(u,v)\in \hat{\mathcal{D}}_{m}$, and $\hat{q}^{(s)}_{u,v} = 0$ otherwise \; 

 Reconstruct sensitive attributes $\hat{S}$ for $\mathcal{D}_{m}$, getting $\hat{\mathcal{D}}_{m}$, and for each sample $(u,v)\in \hat{\mathcal{D}}_{m}$, compute $\hat{q}^{(s)}_{u,v}$  (=$1/|\hat{\mathcal{D}}_{m}|$ if $\hat{S}=s$ else 0), obtaining $\hat{Q}^{(s)}=\{ \hat{q}^{(s)}_{u,v} | (u,v) \in \hat{\mathcal{D}}_{m} \}$ \;

% \frac{1}

% {|\hat{\mathcal{D}}_{m}^{(s)}|}$ if $\hat{S}=s$ for each sample $(u,v)\in \hat{\mathcal{D}}_{m}$, and $\hat{q}^{(s)}_{u,v} = 0$ otherwise, obtaining $\hat{Q}^{(s)}=\{ \hat{q}^{(s)}_{u,v} | (u,v) \in \hat{\mathcal{D}}_{m} \}$ \; 

%  Reconstruct sensitive attributes $\hat{S}$ for $\mathcal{D}_{m}$, getting $\hat{\mathcal{D}}_{m}$, and $\hat{Q}^{(s)} \in \mathbb{R}^{|\hat{\mathcal{D}}_{m}|}$ with $i$-th entry
% $\hat{p}^{(s)}_{m,i} = \frac{1}{|\hat{\mathcal{D}}_{m}^{(s)}|}$ if the $i$-th reconstructed sensitive attribute $\hat{S}=s$, and 0 otherwise \; 
 Estimate $\rho_{s}\approx  P(S\ne \hat{S}|S=s)$ with the validation set held out by $\mathcal{D}_{k}$\;
 
% 	\KwOut{Predictions to user-item
    % \zy{Initialize $\phi_{s}$ and $\{\phi_{t}\}_{t}$}\; 
    %(with Gaussian distribution if they represent embeddings, else constant values)
     
     \For{$t=1,\dots,T$}{
     Update $\theta$ using gradient descent: $\theta \leftarrow \theta - \alpha_{\theta} \nabla_{\theta}{L\left (\theta \right )} - \sum_{s=0}^{1}\lambda_{s}\alpha_{\theta}  \nabla_{\theta} L_{s}( \theta,\tilde{Q}^{(s)})$\;
     % \left| \eta_{k}^{(s)}\mathbb{E}_{\mathcal{D}_{k}^{(s)}}\left [\hat{R}\right] + \eta_{m}^{(s)}\sum_{(u_{i},v_{i})\in \mathcal{D}_{m}} \tilde{p}^{(s)}_{L,i} \hat{r}_{u_i,v_i} - \mathbb{E}_{\mathcal{D}}\left [\hat{R}\right]\right| $\;
     %\sum\limits_{s=0}^{1}\lambda_{s} \alpha_{\theta} \nabla_{\theta} \left| \frac{\eta_{k}^{(s)}}{\left|\mathcal{D}^{(s)}_{k}\right|}\sum\limits_{i=1}^{\left|\mathcal{D}^{(s)}_{k}\right|}\hat{r}_{u_i,v_i} + \eta_{m}^{(s)}\sum\limits_{i=1}^{n_L} \tilde{p}^{(s),t}_{L,i} \hat{r}_{u_i,v_i} - \frac{1}{n}\sum\limits_{i=1}^{n} \hat{r}_{u_i,v_i}\right|$ \;
\For{$s=0,1$}{
Update $\tilde{Q}^{(s)}$ using gradient ascent:
$\tilde{Q}^{(s)} \leftarrow \tilde{Q}^{(s)} + \sum\limits_{s=0}^{1}\lambda_{s} \alpha_{{q}} \nabla_{\tilde{Q}^{(s)}} 
L_{s}( \theta,\tilde{Q}^{(s)})$ \;

% \left| \eta_{k}^{(s)}\mathbb{E}_{\mathcal{D}_{k}^{(s)}}\left [\hat{R}\right] + \eta_{m}^{(s)}\sum_{(u_{i},v_{i})\in \mathcal{D}_{m}} \tilde{p}^{(s)}_{L,i} \hat{r}_{u_i,v_i} - \mathbb{E}_{\mathcal{D}}\left [\hat{R}\right]\right|$ \;
Project $\tilde{Q}^{(s)}$ onto $\mathbb{B}(\rho_{s};\hat{Q}^{(s)})$\;
% the intersection of  $\left \| \tilde{p}^{(s),t+1}- \hat{p}^{(s)} \right \| _{1} \leq 2 \rho_{s}$ and the simplex $\left \| \tilde{p}^{(s),t+1} \right \| _{1} = 1$ \;
} 
% \\
}
    % Randomly sample a batch of data from $\mathcal{D}$\;
    % 	Normally update model parameters except $\phi_{s}$ and $\{\phi_{t}\}_{t}$, and then keep them fixed\;
    % 	Randomly sample an environment $t\in \{1,\dots,T\}$ \;
    % 	// Update $\phi_{s}$,  fixing all $\phi_{t}$\;
    % 	Compute $\widetilde{\phi}_{s}$ with Equation~\eqref{eq:meta-training}\;
    % 	Update $\phi_{s}$ according to Equation~\eqref{eq:meta-test}\;
    % 	// Update $\phi_{t}$, fixing $\phi_{s}$\;
    % 	Update $\phi_{t}$ according to Equation~\eqref{eq:update_t}\; 
%\vspace{-3pt}
% \end{algorithmic}
\end{algorithm}
\setlength{\textfloatsep}{0.1cm}
% Important

\vspace{3pt}
\noindent \textbf{Empirical Lagrangian form.} 
The optimization problem in Equation~\eqref{eq:fair-emp-constraint} is difficult to solve because it involves fairness constraints for all potential distributions $\tilde{Q}^{(s)}$ in $\mathbb{B}(\rho_{s};\hat{Q}^{(s)})$. To overcome this challenge, we convert the fairness constraints into regularizers and optimize the worst-case unfairness within $\mathbb{B}(\rho_{s};\hat{Q}^{(s)})$ to ensure fairness for the entire set~\cite{dro-rawlfair}. Finally, we reformulate Equation~\eqref{eq:fair-emp-constraint} as a min-max optimization problem, given by:
% \sth{By converting the constraint terms in Equation~\eqref{eq:constraint-fair-in-dist} into the empirical Lagrangian terms and
% minimizing the worst-case unfairness within the ambiguity set, the optimization problem can be solved with a min-max game:}

% Then, following the conventional DRO (from Equation~\eqref{eq:fair-constraint-dro} to Equation~\eqref{eq:fair-minmax-dro}), we further convert the problem 
% into the empirical Lagrangian form as follows:
% Then the problem in Equation~\label{eq:fair-emp-constraint} could be easily converted 
% \begin{equation}\small
% \label{eq:fair-emp-constraint}
% \begin{aligned}
% \underset{\theta}{\text{min}} \max\limits_{\substack{\tilde{P}_{m}^{(s)} \in \mathbb{B}({\rho_{s}};\hat{Q}^{(s)}) \\s=0,1}} L_{MF}\left (\theta \right ) + \sum_{s=0}^{1}\lambda_{s}  \left| \eta_{k}^{(s)}\mathbb{E}_{\mathcal{D}_{k}^{(s)}}\left [\hat{R}\right] \right . \\ \left . + \eta_{m}^{(s)}\sum_{(u_{i},v_{i})\in \mathcal{D}_{m}} \tilde{p}^{(s)}_{L,i} \hat{r}_{u_i,v_i} - \mathbb{E}_{\mathcal{D}}\left [\hat{R}\right]\right|  
% \end{aligned},
% \end{equation}
\begin{equation}
\small
\label{eq:emp-lag-form}
% \vspace{-3pt}
% \begin{aligned}
% \underset{\theta}{\text{min}} \max\limits_{\substack{\tilde{Q}^{(s)} \in \mathbb{B}({\rho_{s}};\hat{Q}^{(s)}) \\s=0,1}} L\left (\theta \right ) + \sum_{s=0}^{1}\lambda_{s} L_{s}( \theta,\tilde{Q}^{(s)}),
\underset{\theta}{\text{min}} \max\limits_{\substack{\tilde{Q}^{(s)} \in \mathbb{B}({\rho_{s}};\hat{Q}^{(s)}),~s=0,1}} L\left (\theta \right ) + \sum_{s}\lambda_{s} L_{s}( \theta,\tilde{Q}^{(s)}),
% \sum_{s=0}^{1}\lambda_{s}  \left| \eta_{k}^{(s)}\mathbb{E}_{\mathcal{D}_{k}^{(s)}}\left [\hat{R}\right] \right . \\ \left . + \eta_{m}^{(s)}\sum_{(u_{i},v_{i})\in \mathcal{D}_{m}} \tilde{p}^{(s)}_{L,i} \hat{r}_{u_i,v_i} - \mathbb{E}_{\mathcal{D}}\left [\hat{R}\right]\right|  
% L_{s}(\theta)= \left| \eta_{k}^{(s)}\mathbb{E}_{\mathcal{D}_{k}^{(s)}}\left [\hat{R}\right]  + \eta_{m}^{(s)}\sum_{(u_{i},v_{i})\in \mathcal{D}_{m}} \tilde{p}^{(s)}_{L,i} \hat{r}_{u_i,v_i} - \mathbb{E}_{\mathcal{D}}\left [\hat{R}\right]\right|  
% \end{aligned}
% \vspace{-3pt}
\end{equation}
where $\lambda_{s}$ is the hyper-parameter to control the strength of the fairness regularizer $L_{s}(\theta,\tilde{Q}^{(s)})$ obtained via Lagrangian trick: %, which is given by
%is the empirical lagrangian term to promote fairness, which is computed as follows:
%where $\lambda_{s}$ is the hyper-parameter to control the strength of the fairness regularization term and $L_{s}(\theta,\tilde{Q}^{(s)})$ is the empirical lagrangian term to promote fairness, which is computed as follows:
\begin{equation}
    \small
    \vspace{2pt}
    \label{eq:l_s}
    L_{s}(\theta,\tilde{Q}^{(s)})= \left| \eta_{k}^{(s)}\mathbb{E}_{\mathcal{D}_{k}^{(s)}}\left [\hat{R}\right]  + \eta_{m}^{(s)}\mathbb{E}_{\tilde{Q}^{(s)},\hat{\mathcal{D}}_{m}}\left[\hat{R}\right]
    - \mathbb{E}_{\mathcal{D}}\left [\hat{R}\right]\right|.
    \vspace{2pt}
\end{equation}
Here, $\tilde{Q}^{(s)}$ becomes learnable parameters. We update it in $\mathbb{B}(\rho_{s};\hat{Q}^{(s)})$ via `max' optimization to find the worst case for fairness.

\vspace{3pt}
\noindent \textbf{Learning algorithm}. 
% To solve the optimization  problem in Equation~\eqref{eq:fair-emp-constraint}, we can utilize the algorithm for solving the DRO problem provided in~\cite{gda}, with the pseudocode provided in Algorithm~\ref{alg:DRFO}. At each iteration, the algorithm first updates the recommender model parameters $\theta$ to minimize the loss in Equation~\eqref{eq:fair-emp-constraint} (line 5); then, for each sensitive attribute $s \in \left\{ 0,1 \right\}$, it updates $\tilde{Q}^{(s)}$ to maximize the loss in Equation~\ref{eq:fair-emp-constraint} (line 7), and projects the updated result onto $\mathbb{B}(\rho_{s};\hat{Q}^{(s)})$ (line 8) to obtain final $\tilde{Q}^{(s)}$ for this iteration. 
To solve the optimization problem in Equation~\eqref{eq:emp-lag-form}, we use the algorithm in~\cite{gda} for solving DRO, which is summarized in Algorithm~\ref{alg:DRFO}. At each iteration of the algorithm, we first update the model parameters $\theta$ to minimize the total loss in Equation~\eqref{eq:emp-lag-form} (line 5). Then, for each sensitive attribute value $s \in \{0,1 \}$, we update $\tilde{Q}^{(s)}$ to maximize the fairness-related loss $L_{s}(\theta,\tilde{Q}^{(s)})$ in Equation~\eqref{eq:emp-lag-form} (line 7), and project the updated result onto $\mathbb{B}(\rho_{s};\hat{Q}^{(s)})$ (line 8, \textcolor{black}{avoiding $\tilde{Q}^{(s)}$ beyond the set}) to obtain the final $\tilde{Q}^{(s)}$ for that iteration.  

\subsection{Discussion}
\label{sec:discussion}

We now discuss the extension of DRFO for situations where certain users are not willing to reconstruct their sensitive attributes. 
The main challenge lies in the inability to reconstruct sensitive attributes, which is essential for building the ambiguity set in DRFO.
% in that we cannot reconstruct the users' sensitive attributes to build the ambiguity set needed in DRFO.
% However, even with a random distribution, we can build a broad ambiguity set that covers the true distribution of these users by using a large robust radius.
% Thereafter, by optimizing the worst-case unfairness with the broad ambiguity set, we could also achieve fairness in these situations. 
% Nevertheless, we can build a broad ambiguity set that covers the true distribution of these users by using a large robust radius.
However, even with a random distribution, we can build a broad ambiguity set that covers the true distribution of these users by using a large robust radius.
Thereafter, by optimizing the worst-case unfairness with the broad ambiguity set, we could still achieve fairness. 
Let $\mathcal{D}_{b}$ denote the interaction data of such users.
To implement this approach, we first \textit{randomly} assign sensitive attributes for users in $\mathcal{D}_{b}$, serving as reconstructed sensitive attribute $\hat{S}$ in DRFO and obtaining $\hat{\mathcal{D}}_{b}$. Then, we define empirical $\hat{Q}_{b}^{(s)}$, $\tilde{Q}_{b}^{(s)}$, and $\mathbb{B}(\rho_{s}^{\prime};\hat{Q}_{b}^{(s)})$ for $\hat{\mathcal{D}}_{b}$, referring to those of $\hat{\mathcal{D}}_{m}$. To build the broad ambiguity set, we set the robust radius $\rho_{s}^{\prime}$ to 1.
For users with missing sensitive attributes and permitting reconstruction, we use $\mathcal{D}_{r}$ to denote their interactions ($\mathcal{D}_{r} \cup \mathcal{D}_{b} = \mathcal{D}_{m}$, $\mathcal{D}_{r} \cap  \mathcal{D}_{b} = \emptyset$). Then we define $\mathcal{\hat{D}}_{r}$, $\hat{Q}^{(s)}_{r}$, $\tilde{Q}^{(s)}_{r}$, and $\mathbb{B}(\rho_{s};\hat{Q}_{r}^{(s)})$ for $\mathcal{D}_{r}$ in the same way as that of $D_{m}$ described in Section~\ref{sec:DRFO}.
% With a slight abuse of notation, we still use $\mathcal{D}_$ to denote data containing users with missing sensitive attributes and permitting reconstruction. Then we define $\mathcal{\hat{D}}_m$, $\hat{Q}^{(s)}$, $\tilde{Q}^{(s)}$, and $\mathbb{B}(\rho_{s}^{\prime};\hat{Q}_{b}^{(s)})$ in the same way as that described in Section~\ref{sec:DRFO}.
Finally, we could optimize the fairness for $\hat{\mathcal{D}}_{b}$ and $\hat{\mathcal{D}}_{r}$ in a similar way to $\hat{\mathcal{D}}_{m}$ using DRO. Adding such optimization parts into Equation~\eqref{eq:fair-emp-constraint}, we obtain the final optimization problem:
\begin{equation}
% \vspace{-30pt}
\small
\label{eq:fair-emp-constraint-refuse-users}
\begin{aligned}
&\underset{\theta}{\text{min}} \max\limits_{\substack{\tilde{Q}^{(s)}_{r} \in \mathbb{B}(\rho_{s};\hat{Q}^{(s)}_{r})\\
\tilde{Q}_{b}^{(s)} \in \mathbb{B}({\rho_{s}^{\prime}}; \hat{Q}_{b}^{(s)})\\s=0,1}} 
L\left (\theta \right ) +   \sum_{s=0}^{1}\lambda_{s} \bigg| \eta_{k}^{(s)}\mathbb{E}_{\mathcal{D}_{k}^{(s)}}\left [\hat{R}\right] \\ & + 
\eta_{r}^{(s)}\mathbb{E}_{\tilde{Q}^{(s)}_{r},\hat{\mathcal{D}}_{r}}\left[\hat{R}\right] + \eta_{b}^{(s)}\mathbb{E}_{\tilde{Q}^{(s)}_{b}, \hat{\mathcal{D}}_{b}}\left[\hat{R}\right] - \mathbb{E}_{\mathcal{D}}\left [\hat{R}\right]\bigg| , 
\end{aligned}
\vspace{10pt}
\end{equation}
%(u_{i},v_{i})\in \mathcal{D}_{m}(u_{i},v_{i})\in \mathcal{D}_{b}
where $\mathbb{E}_{\tilde{Q}^{(s)}_{r},\hat{\mathcal{D}}_{r}}\left[\hat{R}\right]$ and $\mathbb{E}_{\tilde{Q}^{(s)}_{b},\hat{\mathcal{D}}_{b}}\left[\hat{R}\right]$ represents the empirical expectation of $\hat{R}$ under the distribution $\tilde{Q}^{(s)}_{r}$ and $\tilde{Q}^{(s)}_{b}$, respectively, computed similarly to Equation~\eqref{eq:empirical constraint}, and
$
\eta_{k}^{(s)}=\frac{ |\mathcal{D}^{(s)}_{k}|}{|\mathcal{D}_{k}^{(s)}| + |\hat{\mathcal{D}}_{r}^{(s)}| + |\hat{\mathcal{D}}_{b}^{(s)}|},$
$\eta_{r}^{(s)}=\frac{ |\hat{\mathcal{D}}_{r}^{(s)}|}{ |\mathcal{D}_{k}^{(s)} |+ |\hat{\mathcal{D}}_{r}^{(s)}| + |\hat{\mathcal{D}}_{b}^{(s)}|},$ and $ \eta_{b}^{(s)}=\frac{|\hat{\mathcal{D}}_{b}^{(s)}|}{ |\mathcal{D}_{k}^{(s)}| +  |\hat{\mathcal{D}}_{r}^{(s)}| + |\hat{\mathcal{D}}_{b}^{(s)}|}.
$
%Regarding $\rho_{s}^{\prime}$, we set it as $1$, which serves as an upper bound for TV distance (see Equation~\eqref{eq:def-tv-original} in Appendix~\ref{sec:appendix-tv}). 
The learning algorithm can be formulated following the structure outlined in Algorithm~\ref{alg:DRFO}.

\section{Experiments}
In this section, we conduct extensive experiments to answer the following research questions:
\textbf{RQ1:} Can DRFO provide fair recommendations with limited sensitive attributes? \textbf{RQ2:} How do reconstruction errors affect the effectiveness of DRFO in ensuring fairness? What is the performance of DRFO in terms of fairness on both users with known sensitive attributes and those with unknown sensitive attributes? \textbf{RQ3:} Can DRFO ensure fairness when some users prohibit the reconstruction of sensitive attributes?

\begin{table}[t]
\centering
\caption{Statistics of the evaluation datasets.
% `\#Users' denotes the number of users.
% `S=0' denotes the female group, and `S=1' denotes the male group.
}
\vspace{-4pt}
\label{tab:statistics}
\renewcommand\arraystretch{0.9}
	\setlength{\tabcolsep}{1.2mm}{
		\resizebox{0.45\textwidth}{!}{
  \small
\begin{tabular}{c|ccccccc}
\hline
\multirow{2}{*}{Dataset} &  \multirow{2}{*}{\#Items} & \multicolumn{2}{c}{\#Users} & \multicolumn{2}{c}{\#Interactions} & \multicolumn{2}{c}{Mean Ratings} \\ \cline{3-8} 
                         &         & $S=0$        & $S=1$        & $S=0$          & $S=1$          & $S=0$            & $S=1$           \\ \hline
ML-1M                    & 3,244   & 1,153        & 3,144        & 228,191        & 713,590        & 0.5866           & 0.5661          \\ \hline
Tenrec                   & 14,539  & 3,108        & 2,299        & 308,217        & 337,958        & 0.4849           & 0.4676          \\ \hline
\end{tabular}
}}
% \vspace{-0.5cm}
\end{table}
\vspace{-2pt}
\subsection{Experimental Settings}
% \vspace{-2pt}
\subsubsection{Datasets} We conduct experiments on two open benchmark datasets: MovieLens 1M\footnote{\url{https://grouplens.org/datasets/movielens/}}~\cite{harper2015movielens} and Tenrec\footnote{\url{https://static.qblv.qq.com/qblv/h5/algo-frontend/tenrec_dataset.html}}~\cite{yuan2022tenrec}.
\begin{itemize}[leftmargin=*]
    \item \textbf{MoviveLens 1M (ML-1M)}: 
    This is a movie recommendation dataset including user ratings of movies on a scale of $1\text{-}5$ and sensitive user attributes such as `gender'. We select `gender' as the sensitive attribute and transform the rating values into binary labels using a threshold of $3$, where ratings greater than three are labeled as positive (`1') and the rest as negative (`0').
    %This is a movie recommendation dataset provided by GroupLens Research, which includes user ratings of movies on a scale of $1\text{-}5$, as well as sensitive user attributes such as `gender' and `age'. We focus on the `gender' attribute as the sensitive attribute of interest and transform the rating values into binary labels (`0' or `1') using a threshold of $3$, where ratings greater than three are labeled as positive (`1') and the rest as negative (`0').
    
    \item \textbf{Tenrec}: 
    % This is a dataset that includes user feedback such as `like' on articles or videos collected from two different feed recommendation platforms of Tencent. It also includes  anonymized user information including sensitive attributes. 
    % For our experiments, we use the `QB-video' sub-dataset. 
    % We designate `like' as the rating label and consider `gender' as the sensitive attribute being studied. 
    This dataset contains user feedback on articles or videos alongside anonymized user information from Tencent's recommendation platforms. 
    We focus on the `QB-video' sub-dataset for our experiments, using the `like' feedback as the rating label and considering `gender' as the sensitive attribute.
%%%%%%%%%old
    % This is a dataset collected from two feed recommendation platforms operated by Tencent, which contains user feedback on articles or videos alongside anonymized user information, including sensitive attributes. 
    % We focus on the `QB-video' sub-dataset for our experiments, using the `like' feedback as the rating label and considering `gender' as the sensitive attribute under investigation.

    %By adopting this approach, we aim to provide a more nuanced understanding of the relationship between gender and user feedback and its implications for the design of recommendation systems.

\end{itemize}

In this work, we adopt the approach proposed in~\cite{weinsberg2012blurme} to reconstruct sensitive attributes only using user historical interactions. To ensure the effectiveness of the reconstruction, it is crucial for users to have a sufficient number of historical interactions in the two datasets. Therefore, we apply 50-core filtering~\cite{50core, 50-core-2} to select users with more than 50 interactions while performing 10-core filtering for items. Table~\ref{tab:statistics} presents the statistical details of the processed datasets after filtering. We partition the datasets into training, validation, and testing sets using a ratio of 0.7:0.15:0.15.

\begin{table*}[t]
\renewcommand\arraystretch{0.9}
\caption{
Comparison between baselines and DRFO on ML-1M and Tenrec \wrt the fairness metric DP and recommendation performance metric RMSE under varying known sensitive attribute ratios. 
Lower DP and RMSE indicate better results.
}
\vspace{-0.3cm}
\label{tab:overall}
\setlength{\tabcolsep}{1.2mm}{
		\resizebox{0.8\textwidth}{!}{
			\small
\begin{tabular}{lcccccc|ccccc}
\hline
\multirow{2}{*}{Dataset} & \multirow{2}{*}{Model} & \multicolumn{5}{c|}{DP $\downarrow$}                                                   & \multicolumn{5}{c}{RMSE $\downarrow$}                          \\ \cline{3-12} 
                         &                        & 10\%            & 30\%            & 50\%            & 70\%            & 90\%            & 10\%   & 30\%   & 50\%                       & 70\%   & 90\%   \\ \hline
\multirow{6}{*}{ML-1M}   & Basic MF               & 0.0225          & 0.0225          & 0.0225          & 0.0225          & 0.0225          & 0.4147 & 0.4147 & 0.4147                     & 0.4147 & 0.4147 \\
                         & Oracle                 & 0.0009          & 0.0009          & 0.0009          & 0.0009          & 0.0009          & 0.4148 & 0.4148 & 0.4148                     & 0.4148 & 0.4148 \\ \cline{2-12} 
                         & RegK                   & 0.0159          & 0.0132          & 0.0058          & 0.0014          & 0.0013          & 0.4144 & 0.4147 & 0.4147                     & 0.4150 & 0.4149 \\
                         & FLrSA                   & 0.0095          & 0.0119          & 0.0038          & 0.0035          & \textbf{0.0003} & 0.4147 & 0.4147 & 0.4148                     & 0.4151 & 0.4147 \\
                         & CGL                    & 0.0082          & 0.0145          & 0.0056          & 0.0035          & 0.0005          & 0.4147 & 0.4144 & 0.4148                     & 0.4150 & 0.4149 \\
                         & DRFO                & \textbf{0.0034} & \textbf{0.0042} & \textbf{0.0011} & \textbf{0.0013} & 0.0006          & 0.4201 & 0.4201 & 0.4191                     & 0.4223 & 0.4198 \\ \hline
\multirow{6}{*}{Tenrec}  & Basic MF               & 0.0068          & 0.0068          & 0.0068          & 0.0068          & 0.0068          & 0.4503 & 0.4503 & 0.4503                     & 0.4503 & 0.4503 \\
                         & Oracle                 & 0.0001          & 0.0001          & 0.0001          & 0.0001          & 0.0001          & 0.4504 & 0.4504 & 0.4504                     & 0.4504 & 0.4504 \\ \cline{2-12} 
                         & RegK                   & 0.0053          & 0.0052          & 0.0053          & 0.0040          & 0.0051          & 0.4558 & 0.4520 & 0.4530                     & 0.4500 & 0.4500 \\
                         & FLrSA                   & 0.0073          & 0.0073          & 0.0073          & 0.0032          & 0.0013          & 0.4503 & 0.4503 & 0.4503                     & 0.4597 & 0.4594 \\
                         & CGL                    & 0.0073          & 0.0073          & 0.0074          & 0.0029          & 0.0009          & 0.4503 & 0.4503 & 0.4503                     & 0.4518 & 0.4596 \\
                         & DRFO                & \textbf{0.0019} & \textbf{0.0018} & \textbf{0.0013} & \textbf{0.0018} & \textbf{0.0009} & 0.4590 & 0.4575 & \multicolumn{1}{l}{0.4565} & 0.4595 & 0.4596 \\ \hline
\end{tabular}
}}
\vspace{-0.3cm}
\end{table*}

\subsubsection{Compared Methods} 
Achieving fair recommendations with limited sensitive attributes is a novel research problem, and no existing method is specifically designed to address it in recommender systems (to our knowledge). To establish baselines, we select two intuitive methods (RegK and FLrSA) and a method (CGL) proposed in other fields for solving this problem.

%The investigation of achieving fairness with limited sensitive attributes in the context of recommendation systems is a novel research problem, and no dedicated fair recommendation methods have been specifically designed for it. Consequently, we consider comparing the proposed DRFO approach against two intuitive methods (RegK and FLrSA), as well as the CGL method proposed in other fields. As constraining the first moment across different groups is an effective method to promote demographic parity~\cite{recoindependence}, the three baselines are all devised using regularization-based methods as described in Equation~\eqref{eq:mf-fair-reg}\footnote{Regularization-based method is a prominent method for achieving DP metric in recommender systems, while other methods, such as adversarial learning-based methods, do not explicitly consider DP metric for evaluation. Therefore, the baselines are developed using regularization for this study.}.

% 针对推荐系统中的dp公平指标，正则项方法是主流且易于实现的.另外，其他方法如基于对抗的方法不比较该性能指标，所以这里采用基于正则的方法设计baseline
%而且约束不同组间一阶矩提升dp是一个有效的方法。
% 这里采用的三个baseline都是基于equation~中的正则方法设计的，约束一阶矩以提升dp是一个易于理解且有效的方法。
%针对推荐系统中的dp公平指标，正则项方法是主流且易于实现的，其他方法如基于对抗的方法不比较该性能指标。这里，我们全部采用
% The studied problem of how to achieve fairness with limited sensitive attributes is new in recommendation, and there are no specifically designed fair recommendation methods. We thus consider comparing the proposed DRFO with two intuitive methods (RegK and FLrSA), and one method (CGL) proposed in other fields:
\begin{itemize}[leftmargin=*] 

\item[--] \textbf{RegK}~\cite{yao2017beyond} is a regularization-based approach that applies fair regularization only to users with known sensitive attributes. 

\item[--] \textbf{FLrSA}~\cite{yao2017beyond} first reconstructs missing sensitive attributes and then applies the fairness regularization method using both the reconstructed and the known sensitive attributes. 
% This method first reconstructs the sensitive attributes for users with unknown sensitive attributes, and subsequently employs a regularization-based method by considering both the reconstructed sensitive attributes and the known sensitive attributes to achieve fairness.

\item[--] \textbf{CGL}~\cite{CGL} is a representative method for fairness with limited sensitive attributes in computer vision. It reconstructs sensitive attributes and addresses low-confidence reconstructions via random replacement. Subsequently, it employs a regularization-based approach using the obtained attributes. We adapt it to recommender systems.

%%%%%%%%%%old
% \item[--] \textbf{CGL}~\cite{CGL} is a state-of-the-art method, designed to achieve fairness with limited sensitive attributes in computer vision. It reconstructs sensitive attributes and addresses low-confidence reconstructions by random replacement. Then, it employs a regularization-based approach using the obtained sensitive attributes to promote fairness. We adapt this method to the domain of recommender systems.

 % This is a state-of-the-art method designed to achieve fairness with limited annotated sensitive attributes in computer vision.  This method also first reconstructs sensitive attributes, but would further deal with low-confidence reconstructions by random replacement.  Subsequently, a regularization-based approach is employed, utilizing the obtained sensitive attributes to effectively promote fairness. We adapt this method to the realm of recommender systems.
 
%This is a state-of-the-art method designed to achieve fairness with limited annotated sensitive attributes in computer vision. This method focuses on reconstructing sensitive attributes, considering confidence levels and selectively replacing low-confidence reconstructions with random sensitive attributes drawn from the empirical conditional distribution. Subsequently, a regularization-based approach is employed, utilizing the obtained sensitive attributes to effectively promote fairness. We adapt this method to the realm of recommender systems.
\end{itemize}
% todo,为啥只采用这些baseline。推荐公平中对不同公平问题往往采用不同方法，研究其他公平指标的方法不适合在这比较。针对dp指标，这里的约束均值的正则项，是一个易于实现且有效的方法。~recommendation independence. 
% Notably, although RegP and FLrSA are intuitively designed, similar methods have been utilized in other fields to study fairness with limited sensitive attributes~\cite{CGL}.  Meanwhile, for fair comparison, we implement the fairness regularization method according to Equation~\ref{}, for all methods.
% Besides, we additionally compare two methods that do not address the issue of missing sensitive attributes, serving as references:
% Although RegK and FLrSA are intuitively designed, similar methods have been employed in other fields to study fairness with limited sensitive attributes~\cite{CGL}. 
% ReK and FLrSA are intuitively designed, similar to methods used in other fields to study fairness with limited sensitive attributes~\cite{CGL}.
Although RegK and FLrSA are intuitively designed, similar methods have been employed in other fields to study fairness with limited sensitive attributes~\cite{CGL}. 
Additionally, we include two fairness-unawareness methods as reference points.
% Additionally, we include two methods that do not address the issue of fairness with limited sensitive attributes as reference points.
% To ensure fair comparison, we implement the fair regularization term based on Equation~\eqref{\eqref{eq:mf-fair-reg}} to pursue DP fairness for all methods. Notably, although RegK and FLrSA are intuitively designed, similar methods have been employed in other fields to study fairness with limited sensitive attributes~\cite{CGL}. In addition, we include two methods that do not address the issue of missing sensitive attributes as the reference:
\begin{itemize}[leftmargin=*] 

\item[--] \textbf{Basic MF}~\cite{koren2009matrix} refers to the basic Matrix Factorization (MF) model that is trained without addressing fairness issues.
% This baseline MF method is trained without considering fairness.

\item[--] \textbf{Oracle}~\cite{yao2017beyond} assumes knowledge of all users' sensitive attributes and employs the regularization-based method to achieve fair recommendations. Its results serve as an upper bound on fairness.
\end{itemize}
To ensure a fair comparison, we implement the fairness regularization term based on Equation~\eqref{eq:mf-fair-reg} to pursue demographic parity fairness, using the classical Matrix Factorization model as the backbone recommendation model for all methods.

%对比方法 1.只使用已知敏感属性标签的正则项方法 2.直接利用重建敏感属性的正则项方法 
%%3.cgl（Learning Fair Classifiers with Partially Annotated Group Labels）:一个贴标签的方法+利用贴上的标签加正则项约束: 1.重建出敏感属性标签 2.根据验证集上的分类性能，挑选出一个门限，可以区分出置信度高的预测样本和置信度低的预测样本。3.根据这个门限，在测试集上置信度高的样本用重建出的敏感属性，置信度低的样本用已知部分的先验来随机赋敏感属性（如已知部分男性为p，则以概率为p赋男性，1-p赋女性) 4.利用贴上的标签施加公平约束 （文中证明虽然对低置信度的样本随机贴了标签施加了公平约束，但等价于不考虑这部分样本的公平性）
%%%4.已知全部敏感属性并施加约束(作为上限)

\vspace{-2pt}
\subsubsection{Evaluation Protocol.} 
% To simulate the cases with missing sensitive attributes in experiments, for training and validating, we would randomly keep the sensitive attributes of some users with a ratio of $k \in \{0.1, 0.3, 0.5, 0.7, 0.9\}$, and mask that of other users. 
In our experiments, we simulate scenarios involving unknown sensitive attributes by randomly retaining the sensitive attribute for a subset of users while masking it for others. The retention ratio,  varying in \{0.1, 0.3, 0.5, 0.7, 0.9\}, determines the proportion of users whose sensitive attributes are preserved during both training and validation phases. During testing, however, the sensitive attributes of all users are accessible to evaluate fairness. We use the MAD metric in Equation~\eqref{eq: MAD} to measure fairness (DP) and the root mean squared error (RMSE) to measure recommendation performance. A smaller value for both metrics indicates better performance in terms of fairness or recommendation.
% In our experimental setup, we simulate scenarios involving missing sensitive attributes by randomly retaining sensitive attributes for a subset of users while masking them for others. The retention ratio, denoted as $k \in \{0.1, 0.3, 0.5, 0.7, 0.9\}$, determines the proportion of users for whom sensitive attributes are preserved during both training and validation phases. During testing, however, the sensitive attributes of all users are accessible to evaluate fairness. To measure fairness (demographic parity), we directly utilize the mean absolute deviation metric in Equation~\eqref{eq: MAD}. To measure the recommendation performance, we take the Root Mean Squared Error (RMSE). A smaller value for both metrics indicates better performance regarding fairness or recommendation.

\vspace{-2pt}
\subsubsection{Implementation Details} 
% We implement all methods using the packages DeepCTR\footnote{\url{https://github.com/shenweichen/DeepCTR-Torch}} and PyTorch\footnote{\url{https://pytorch.org/}}. 
For a fair comparison, we optimize all models using the Adam optimizer~\cite{adam} with the default embedding size of 32. Before applying the fairness methods, we pre-train the MF model using grid search to determine the optimal learning rate from the range of $\{1e^{-2}, 1e^{-3} \}$ and the best weight decay from the range of $\{1e^{-1},1e^{-2},...,1e^{-7}\}$. For the fairness models, we initialize their backbone MF model with the pre-trained one and then fine-tune them with a fixed learning rate of $1e^{-3}$. 
% To control the fairness regularization, 
We tune the fairness regularization coefficient in the range of $\{0.01, 0.05, 0.1, 0.5, 1, 5, 10\}$ for the baselines. 
For DRFO, we set the hyper-parameter that controls the strength of fairness constraints to $10$ for ML-1M and $1$ for Tenrec, and set the learning rate $\alpha_q$ for updating $\Tilde{Q}^{(s)}$ to $1e^{-3}$ 
for ML-1M and $1e^{-2}$ for Tenrec. Additionally, for methods that involve reconstructing unknown sensitive attributes, we use the same classifier proposed in~\cite{weinsberg2012blurme} that is trained to fit known sensitive attributes using historical user interactions. 
To ensure a fair comparison, we select hyper-parameters that achieve the best fairness (measured by DP) while also maintaining at least 98\% of the best recommendation performance (measured by RMSE) achieved by the basic MF on the validation set. This approach is commonly used in previous studies to balance the trade-off between fairness and recommendation performance~\cite{recfair-survey, evaluation}.
We release our code at: \url{https://github.com/TianhaoShi2001/DRFO}.
\begin{figure}[t]
    \centering 
    \includegraphics[width=0.47\textwidth]{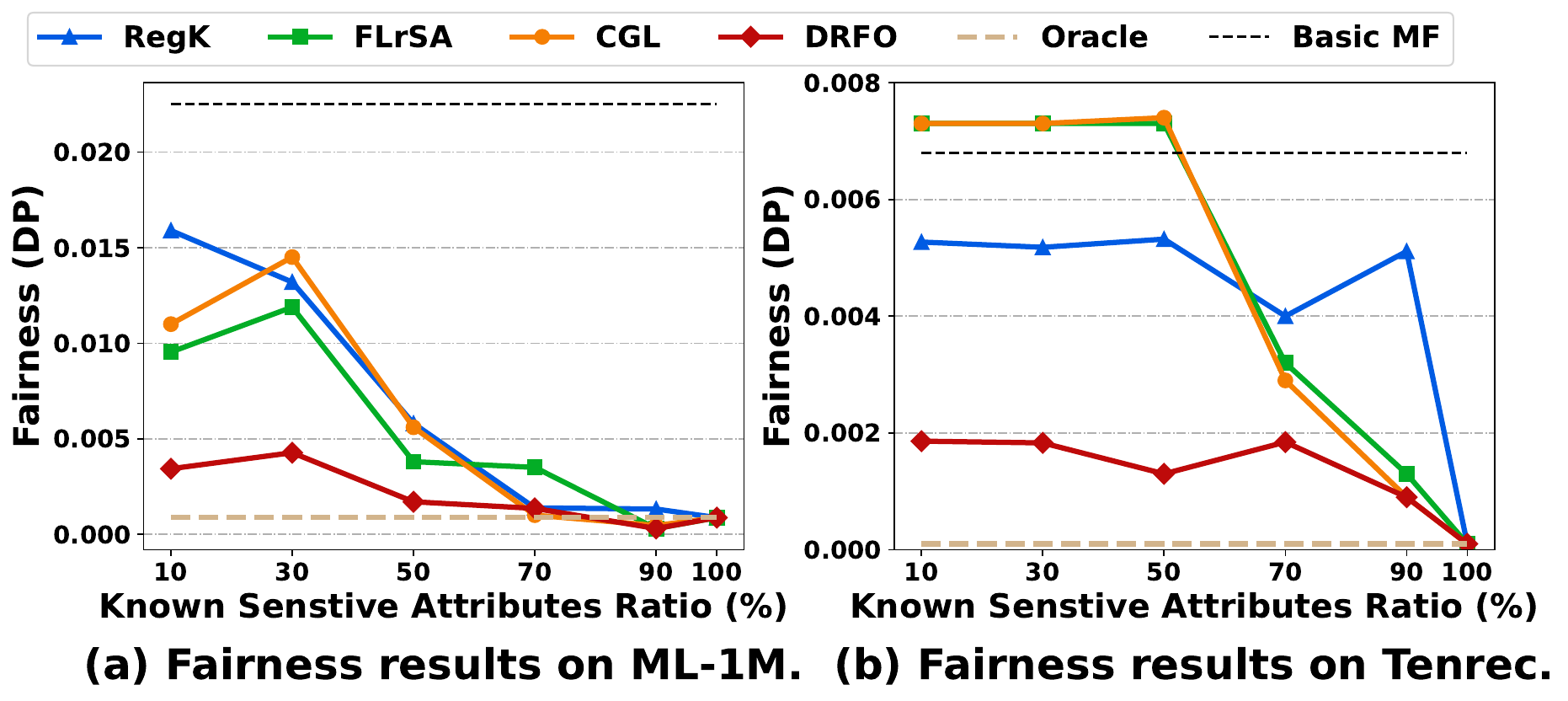}
    \vspace{-0.3cm}
    \caption{Fairness comparison between baselines and DRFO on ML-1M and Tenrec for varying known sensitive attribute ratios. Lower DP values indicate better fairness.
    }
        \label{fig:overall}
    % \vspace{-0.5cm}
    
\end{figure}

\vspace{-2pt}
\subsection{Overall Performance (RQ1)}
% \subsubsection{Fairness with Limited Sensitive Attributes (RQ1).} 
% 1.相比baseline，DRFO可以有效提升公平性。在已知用户比例比较低的时候，baseline方法不能很好提供很好的公平，因为xxx。dro通过优化xx，可以保证公平，在已知用户比例低的时候仍然有效。随着已知用户比例提升，dro仍然可以提供公平的推荐，而baseline方法的公平性也提升，我们认为这
% 3.cgl为啥不好，预测器不能处理，放在与baseline一起说了。（分两段，第一段是dro好，baseline不好，第二段是公平的变化趋势。）

% 2.dro的模型performance由于施加了额外的约束，由于直接作用在预测得分上强迫xxx，这与学习更好的是矛盾的，值得注意的是，我们通过约束方法在验证集上只下降2%的性能的范围内搜参，DRFO以少量performance下降的代价提供了在缺少敏感属性时的公平提升。

We first evaluate methods when all users' sensitive attributes are allowed to be reconstructed. We summarize the results in Table~\ref{tab:overall} for varying proportions of users with known sensitive attributes. 
The corresponding fairness curves for each method are shown in Figure~\ref{fig:overall} to illustrate the impact of the retention ratio.
% To further illustrate the impact of the retention ratio on each method, we also plot the corresponding fairness curves in Figure~\ref{fig:overall}. 
From the table and figure, we have drawn the following observations:
%%%%%%%%%%old
% We first evaluate the performance of compared methods when all users' sensitive attributes are allowed to be reconstructed. We summarize the results in Table~\ref{tab:overall} for varying proportions of users with known sensitive attributes. To further illustrate the impact of the retention ratio on each method, we also plot the corresponding fairness curves in Figure~\ref{fig:overall}. From the table and figure, we have drawn the following observations:

% Table~\ref{tab:overall} and Figure~\ref{fig:overall} summarizes the overall performance comparison between the baselines and DRFO with different known sensitive attribute ratios
% \footnote{We assume explicit consent has been obtained from all users for conducting sensitive attribute reconstruction to pursue fairness in this subsection.}.
% % \footnote{We assume all users consent to conduct sensitive attribute reconstruction for pursuing fairness in the subsection.}. 
% We have the following observations.

% 1.DRFO 公平性在绝大多数情况下都是最优的。
% 2.cgl与FLrSA在tenrec上不如basic mf，这是因为对错误的标签施加约束反而导致损害公平性。另外cgl方法没有超过一般方法，我们认为xxx todo，找点别人的引用
% 3.对于性能指标而言,dro的模型performance由于施加了额外的约束，由于直接作用在预测得分上强迫xxx，这与学习更好的是矛盾的，值得注意的是，我们通过约束方法在验证集上只下降2%的性能的范围内搜参，DRFO以少量performance下降的代价提供了在缺少敏感属性时的公平提升。

% todo 说明了直接利用敏感属性重建器的方法提升公平性的能力有限。

%%%%%%%%%%%%%%%% new 0807
\begin{itemize}[leftmargin=*]
% RegK,regp越小越差---高度依赖假设，无法解决fair,各自原因。dro都很好，但越小越差---因为xxx，但还是提供了。cgl也是越小越差，原因，难以适应，不如dro。
% 比较DRFO与RegK、FLrSA。说明两点 1. 现有方法不行 2.dro很行
% RegK、regp在已知敏感属性比例高时，可以为提供公平的推荐结果，但随着已知敏感属性比例降低，RegK，regp的公平性下降。这是因为，RegK只能对已知敏感属性的用户施加公平约束，对剩余部分用户的公平性没有保证。FLrSA施加有偏的约束，随着未知敏感属性用户比例增多，约束的偏见更大，也无法提供公平的推荐。这验证了现有公平方法高度依赖于假设能获取全部敏感属性的假设。
%FLrSA完全依赖含有重建误差的敏感属性施加有偏的公平约束，随着未知敏感属性用户比例增多，约束的偏见更大，也无法提供公平的推荐。
    \item \sth{Basic MF method exhibits notably inferior fairness compared to regularization-based methods (\eg FLrSA) with a substantial portion ($\ge 90 \% $) of known sensitive attributes, particularly when compared to the Oracle model. This emphasizes the effectiveness of regularization-based approaches in enhancing recommendation fairness.}
    However, \sth{as the proportion of known sensitive attributes decreases ($\leq$50\%), the fairness performance of regularization-based methods rapidly declines, highlighting the importance of addressing limited sensitive attribute issues for improving fairness.}    % Basic MF exhibits inferior fairness compared to the Oracle model, demonstrating the effectiveness of regularization-based approaches in enhancing recommendation fairness.
    % \sth{Additionally, when a substantial portion of sensitive attributes (90\%) is known, regularization-based methods have proven to be highly effective in enhancing fairness.}
    % when the sensitive attributes of partial users are unknown, directly applying regulation-based methods (e.g., RegK) results in significantly greater unfairness than the Oracle model, highlighting the critical importance of addressing limited sensitive attribute issues to improve fairness.
    
    % The fairness of the basic MF model is much worse than that of the Oracle model, indicating that regulation-based fairness methods can effectively promote recommendation fairness. However, when the sensitive attributes of partial users are unknown, directly applying regulation-based methods (e.g., RegK) results in significantly greater unfairness than the Oracle model, highlighting the critical importance of addressing limited sensitive attribute issues to improve fairness.

%%%%%%%%%%% todo 这一段的意思表达有问题
% 高比例更好，说明在高比例下利
    \item 
    When the proportion of known sensitive attributes is small ($\leq$50\%), FLrSA successfully outperforms RegK on ML-1M but fails on Tenrec in fairness. This can be attributed to the fact that reconstructing sensitive attributes for Tenrec is more difficult\footnote{It is supported by the lower reconstruction accuracy (AUC) observed on Tenrec.}, thus suffering more reconstruction errors and invalidating the vanilla reconstruction-based method FLrSA. These findings affirm the significant impact of reconstruction errors on the efficacy of reconstruction-based approaches to improve fairness.

    \item 
    % \sth{Despite taking reconstruction errors into consideration, 
    % % when the proportion of known sensitive attributes is small ($\leq$50\%), 
    % CGL exhibits a comparable behavior to FLrSA.} This may be attributed to CGL's practice of randomly assigning sensitive attributes to samples with low reconstruction confidence, which does not \sth{ensure the resolution of the reconstruction error problem.}
    \sth{
    % Despite taking reconstruction errors into consideration,  CGL exhibits similar to FLrSA in fairness.
    % Despite considering reconstruction errors, 
    Despite taking reconstruction errors into consideration, CGL performs similarly to FLrSA in fairness.
    } This may be due to the fact that CGL randomly assigns sensitive attributes to samples with low reconstruction confidence, 
    which does not \sth{ensure the resolution of the reconstruction error problem.}
    % which does not guarantee to address the error problem.
    % This may be attributed to CGL's practice of randomly assigning sensitive attributes to samples with low reconstruction confidence, which does not \sth{ensure the resolution of the reconstruction error problem.}

    % \item 
    % Despite taking reconstruction errors into consideration, CGL can only consistently outperforms FLrSA and RegK in fairness when a small proportion of users have unknown sensitive attributes. When the proportion is large, CGL exhibits a similar phenomenon to FLrSA when compared to RegK. This may be due to the fact that CGL randomly assigns sensitive attributes to samples with low reconstruction confidence, which does not guarantee to address the error problem.

   \item DRFO consistently achieves a higher level of fairness, compared to all baselines except for Oracle, even when the proportion of users with known sensitive attributes is low. This confirms its effectiveness in addressing reconstruction errors to achieve fairer recommendations. The superiority of the approach can be attributed to its DRO-based fair optimization, which minimizes the worst-case unfairness over a distribution set to achieve fairness for the true sensitive attributes contained in the set.

\item 
% DRFO may exhibit slightly lower recommendation performance compared to baseline methods. This is because the optimization of DRFO prioritizes minimizing the worst-case unfairness in all potential distributions, which makes it more aggressive in its pursuit of fairness, leading to greater performance degradation.
\sth{DRFO achieves slightly lower recommendation performance than baselines due to its aggressive fairness pursuit by minimizing worst-case unfairness across potential distributions.}
Nevertheless, our early stopping strategy selects the fairest model within a maximum 2\% drop in RMSE on the validation, ensuring a low drop (around 2\%) in RMSE on the testing for all methods.
% \sth{Nevertheless, our early stopping strategy ensures a maximum 2\% reduction in the validation set RMSE, translating to about 2\% reduction in the testing set RMSE.}
% Despite this modest sacrifice in recommendation performance, DRFO  has achieved fairness improvements exceeding 100\% in many cases, particularly when dealing with numerous users lacking sensitive attributes.
Despite this modest sacrifice in recommendation performance, 
DRFO improves fairness by over 100\% in many cases, 
%DRFO  has improved fairness by exceeding 100\% in many cases, 
particularly when numerous users lack sensitive attributes.
\end{itemize}

\begin{figure}[t]
     \centering  
     \subfigure[\textbf{Fairness results on ML-1M.}]{
     \begin{minipage}[b]{0.475\textwidth}
     \centering
     % \vspace{-100pt}
\includegraphics[scale=0.28]{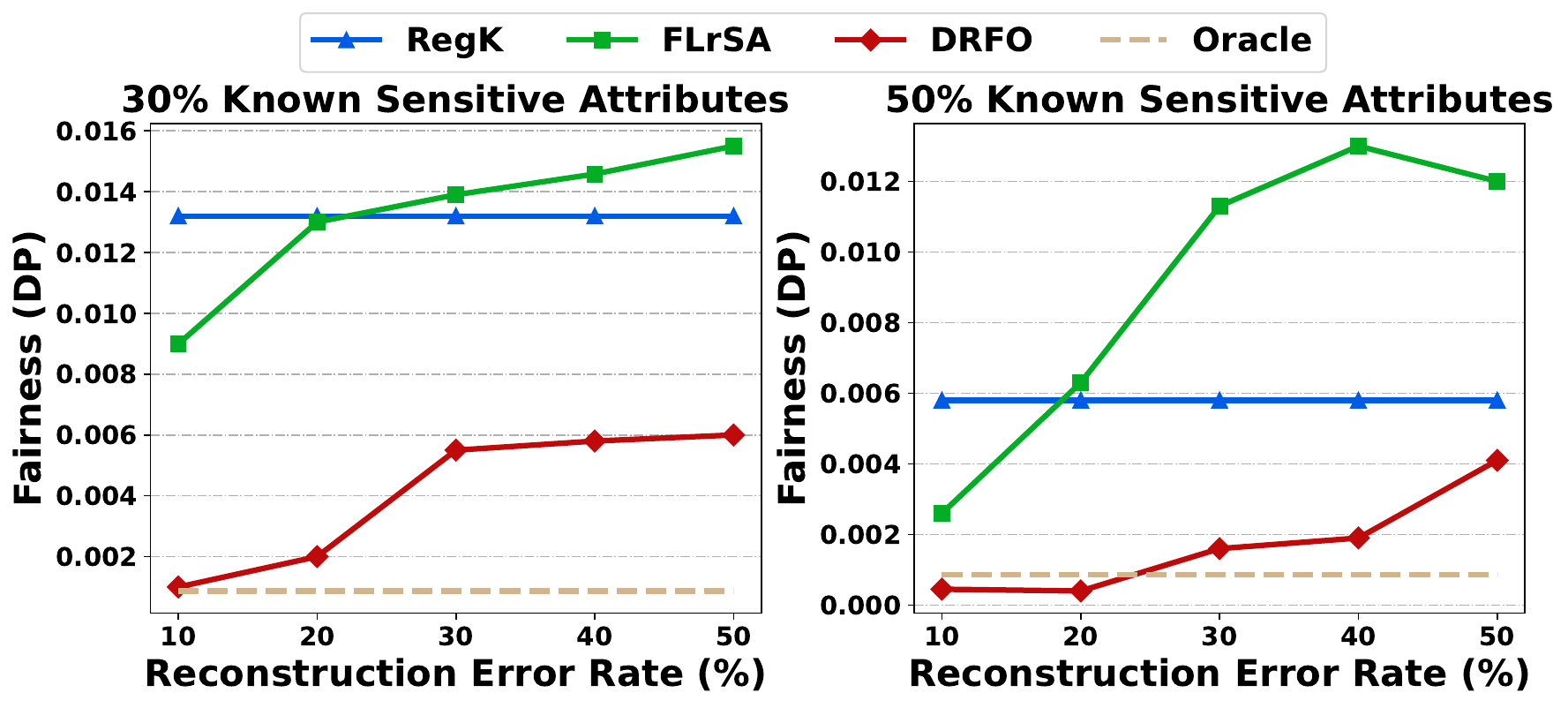}
\vspace{-10pt} % a子标题与a图间距
\end{minipage}
%\caption{fig1}
% \vspace{-0.8cm}
}
% \quad

\vspace{-5pt} %b图与a标题间距
\subfigure[\textbf{Fairness results on Tenrec.}]
 {
\begin{minipage}[b]{0.475\textwidth}
% \vspace{cm}
\centering
\includegraphics[scale=0.28]{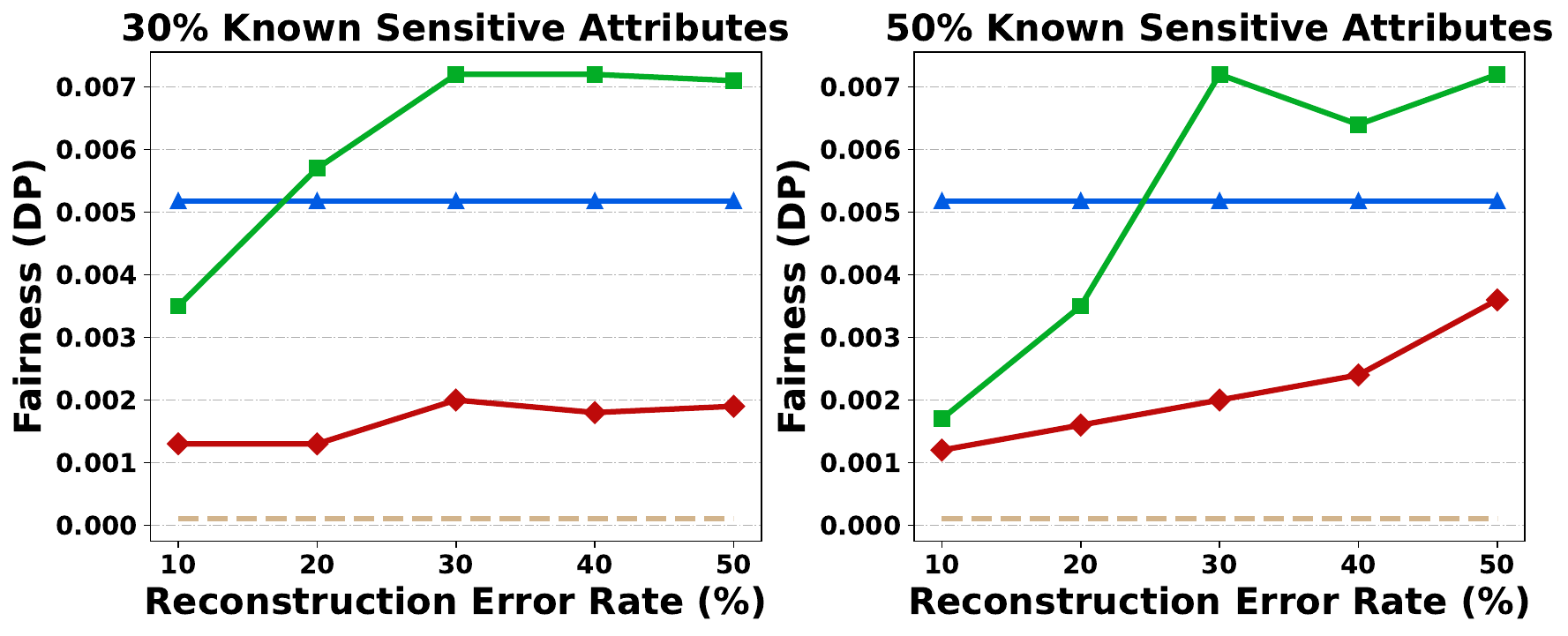}
\vspace{-10pt} % b子标题与图间距
\end{minipage}
}
% \vspace{-0.3cm}
     % \includegraphics[width=0.45\textwidth]{figs/refuse-reconstruction.pdf} 
     \vspace{-10pt} %大标题与上图间距
     \caption{
     Fairness performance under different levels of reconstruction errors for sensitive attributes.
     % Fairness results in scenarios where a portion of users does not allow reconstruction of their attributes among the users with unknown sensitive attributes.
     % 'Ratio' denotes the proportion of users who permit the reconstruction of their sensitive attributes.
     %Comparison of fairness between baselines and DRFO on ML-1M and Tenrec datasets, considering known sensitive attribute ratios of $30\%$ and $50\%$ across various ratios of user consent for sensitive attribute reconstruction (\ie reconstruction sensitive attribute ratio). Lower DP values denote improvements in fairness.
     }
    % \vspace{-0.5cm}
    % \vspace{-15pt}
      \label{fig:reconstruction-noise}
 \end{figure}
% \begin{figure}[t]
%     \centering
%     \includegraphics[width=0.475\textwidth]{figs/main.pdf}
%     \caption{Fairness comparison between baselines and DRFO on ML-1M and Tenrec datasets for varying known sensitive attribute ratios. Lower DP values indicate better fairness performances.
%     }
%     \label{fig:overall}
% \end{figure}

%  \begin{figure}[t]
%      \centering 
% \subfigure[\textbf{Fairness results on ML-1M.}]
% {
% \begin{minipage}[b][0.465\textwidth]
%     \centering
%     \includegraphics[scale=0.275]{figs/reconstruction-noise.pdf}
% \end{minipage}
% }
% \end{figure}
% % \quad
% \subfigure[\textbf{Fairness results on Tenrec.}]
% {
% \begin{minipage}[b][0.465\textwidth]
% \centering
% \includegraphics[scale=0.275]{figs/reconstruction-noise-tenrec.pdf}
% \end{minipage}
% }
%      % \includegraphics[width=0.45\textwidth]{figs/reconstruction-noise.pdf}
%      \vspace{-4pt}
%      \caption{
%      Fairness performance under different levels of reconstruction errors for sensitive attributes.
%      % Comparison of fairness between baselines and DRFO on ML-1M and Tenrec datasets, considering known sensitive attribute ratios of $30\%$ and $50\%$ for different reconstruction error levels. Lower DP values denote improvements in fairness. 
%      }
%      \label{fig:reconstruction-noise}
%      % \vspace{-15pt}
%  \end{figure}

\begin{figure}[t]
     \centering 
{
\includegraphics[height = 3.75cm,width=0.47\textwidth]{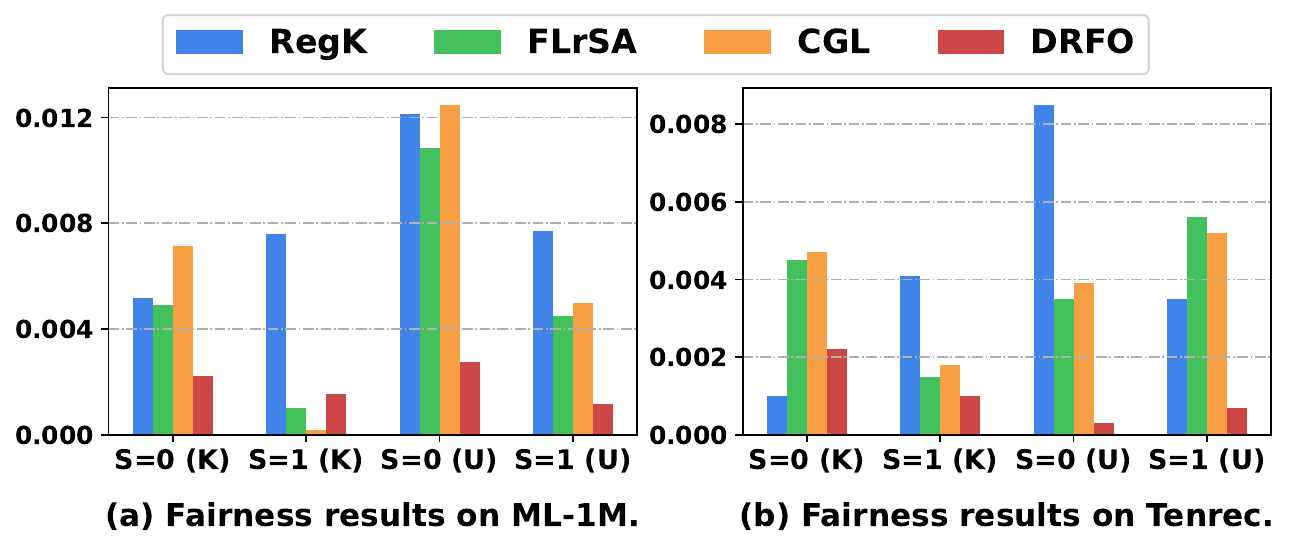}
}  
\vspace{-15pt}
     % \includegraphics[height = 5cm, width=0.5\textwidth]{figs/case-study.pdf}
     % \vspace{-0.15cm}
     \caption{
     Absolute difference of average predicted scores of different groups from global average predictions. Higher difference means more unfairness. `K' stands for `known', and `U' stands for `unknown'. `S=0 (K)' denotes the users with the known sensitive attribute of 0, similarly for others.     
     %Comparison of fairness between baselines and DRFO on ML-1M and Tenrec datasets with a known sensitive attribute ratio of $30\%$ across groups with varying sensitive attributes and different availability of sensitive attributes. Lower Mean Absolute Deviation values denote improvements in fairness.
     }
     \label{fig:case-study}
     % \vspace{-10pt}
      % \vspace{-0.2cm}
 \end{figure}
% \vspace{-2pt}
 \subsection{In-depth Analyses (RQ2)}
We will conduct further experiments to analyze the influence of reconstruction errors on different methods and study the fairness performance of different users.
% \vspace{-2pt}
 \subsubsection{The Effect of Reconstruction Errors}
% In this subsection, we conduct a direct investigation into the impact of reconstruction errors on the ability of DRFO and the baselines to achieve fairness.
\sth{In this subsection, we directly investigate the impact of reconstruction errors on the fairness achievement of DRFO and baselines.
We set the ratio of users with known sensitive attributes to be either 30\% or 50\%, and for the remaining users, we randomly flip the true sensitive attributes of some users to form the reconstructed sensitive attributes, thus introducing reconstruction errors.
The flipping ratio varies from 10\% to 50\%, allowing us to simulate different levels of reconstruction errors\footnote{A 50\% reconstruction error is significant, akin to random guessing, and we flip attributes for males and females at a 1:1 ratio to maintain unchanged gender distributions.}.
% A reconstruction error level of 50\% is already substantial, akin to random guessing. We ensure that the gender distribution remains unchanged by flipping the sensitive attributes of females and males in a 1:1 ratio.}.
}
%% 并且对剩下的用户，我们将真实敏感属性随机反转作为重建的敏感属性，并引入重建误差。
% We set the ratio of users with known sensitive attributes to be either 30\% or 50\%, and for the remaining users, we employ the true attributes as the reconstruction results but introduce reconstruction errors by randomly flipping the attributes of partial users.
% To simulate varying levels of reconstruction error, we adjust the flipping ratio in the range of $\{10\%, 20\%, 30\%, 40\%, 50\%\}$, which results in the same range of reconstruction error rates for unknown parts\footnote{It is noteworthy that a reconstruction error level of 50\% is already substantial, akin to random guessing. We ensure that the gender distribution remains unchanged by flipping the sensitive attributes of females and males in a 1:1 ratio.}.
We compare the fairness of RegK, FLrSA, and DRFO under different reconstruction error levels, excluding CGL due to the lack of reconstruction confidence, which is a necessary condition for CGL. The resulting fairness comparison is presented in Figure~\ref{fig:reconstruction-noise}, where we exclude the recommendation performance as the observed drop is limited to approximately 2\% (we also omit these results in the following experiments for the same reasons). Based on the figure, we make the following observations:

 \begin{itemize}[leftmargin=*]
     % \item As the reconstruction error increases, we observe a decline in the fairness achieved by both DRFO and FLrSA.
     \item
     \sth{As the reconstruction error increases, both DRFO and FLrSA encounter a decline in fairness.}
     Nevertheless, DRFO's smoother curve and consistent superiority over RegK highlight its robustness to reconstruction errors.
     % However, we also observe that the curve of DRFO is flatter than that of FLrSA, and DRFO consistently outperforms RegK, indicating the superior robustness of DRFO to reconstruction errors. 
     The deterioration in DRFO's fairness can be explained by Equation~\eqref{eq: range}, where increases in reconstruction error cause DRFO to select a larger ambiguity set, \sth{ intensifying optimization challenges for achieving fairness.}
     % resulting in an amplified optimization difficulty for DRFO to achieve fairness. 

     \item \sth{Upon surpassing a 20\% reconstruction error, FLrSA no longer surpasses RegK in fairness. Even with just a 10\% reconstruction error, FLrSA fails to match Oracle in fairness. These results emphasize the necessity of addressing reconstruction errors to uphold the reconstruction method's validity.}
 \end{itemize}
\subsubsection{Case Study}
We next investigate whether our DRFO ensures fairness for both users with known sensitive attributes and users with unknown sensitive attributes. 
% \sth{We next investigate whether our DRFO ensures fairness for users with both known and unknown sensitive attributes.}
To achieve this, we propose a metric that measures the absolute difference between a specific group's average predicted rating and the overall average rating, similar to MAD in Equation~\eqref{eq: MAD}. Larger values of this metric indicate more unfairness.
% \sth{To achieve this, we propose a metric akin to MAD in Equation~\eqref{eq: MAD}, which measures the absolute difference between a specific group's average predicted rating and the overall average rating. Larger values of this metric indicate more unfairness.}
We compute this metric for four different groups: users with $S=1$ but $S$ is unknown, users with $S=0$ but $S$ is unknown, users with known $S=1$, and users with known $S=0$. We summarize the result of the compared method in Figure~\ref{fig:case-study}, where 30\% of users have known sensitive attributes. 

%By examining the metric for both known and unknown sensitive attributes, we can assess whether DRFO provides fairness across different groups of users, regardless of whether their sensitive attributes are known or unknown. 

% In this part, we compare the fairness of the DRFO with the baseline method regarding the recommendation outcomes for two types of users: those with known sensitive attributes and those with unknown sensitive attributes. We aim to provide insights into how DRFO ensures fair recommendations when some sensitive attributes are missing.
% Given that demographic parity in Equation~\eqref{eq:dp} can be expressed as $P(\hat{R})=P(\hat{R}|S=s),\ s=0,1$, we adopt a measure that quantifies the unfairness by computing the mean absolute deviation between the average predicted scores of a specific group and the overall average user score.

% As Figure~\ref{fig:case-study} reveals, all the baseline methods demonstrate a substantial deviation between each group and the global average. And the deviation is usually more significant in the groups with unknown sensitive attributes, which suggests a higher degree of unfairness. 
\sth{In Figure~\ref{fig:case-study}, baselines reveal significant absolute differences between 
groups' average and the global average, 
%groups' and the global average, 
particularly noticeable for those with unknown sensitive attributes, implying higher unfairness.}
% These results indicate that the baseline methods are not effective in achieving fair recommendations, particularly for users with unknown sensitive attributes.
These results highlight baselines' inefficacy in achieving fair recommendations, especially for users with unknown attributes.
% In contrast, our proposed DRFO approach exhibits a small absolute difference for all groups, suggesting reduced unfairness.
\sth{In contrast, our proposed DRFO approach exhibits a small unfairness for all groups.}
Furthermore, the unfairness of the groups with known and unknown sensitive attributes is comparable, indicating that our method can successfully overcome the challenge of limited sensitive attributes and provide fair recommendations for users both with known and unknown sensitive attributes.

\begin{figure}[t]
     \centering  
     \subfigure[\textbf{Fairness results on ML-1M.}]{
     \begin{minipage}[b]{0.475\textwidth}
     \centering
     % \vspace{-100pt}
\includegraphics[scale=0.28]{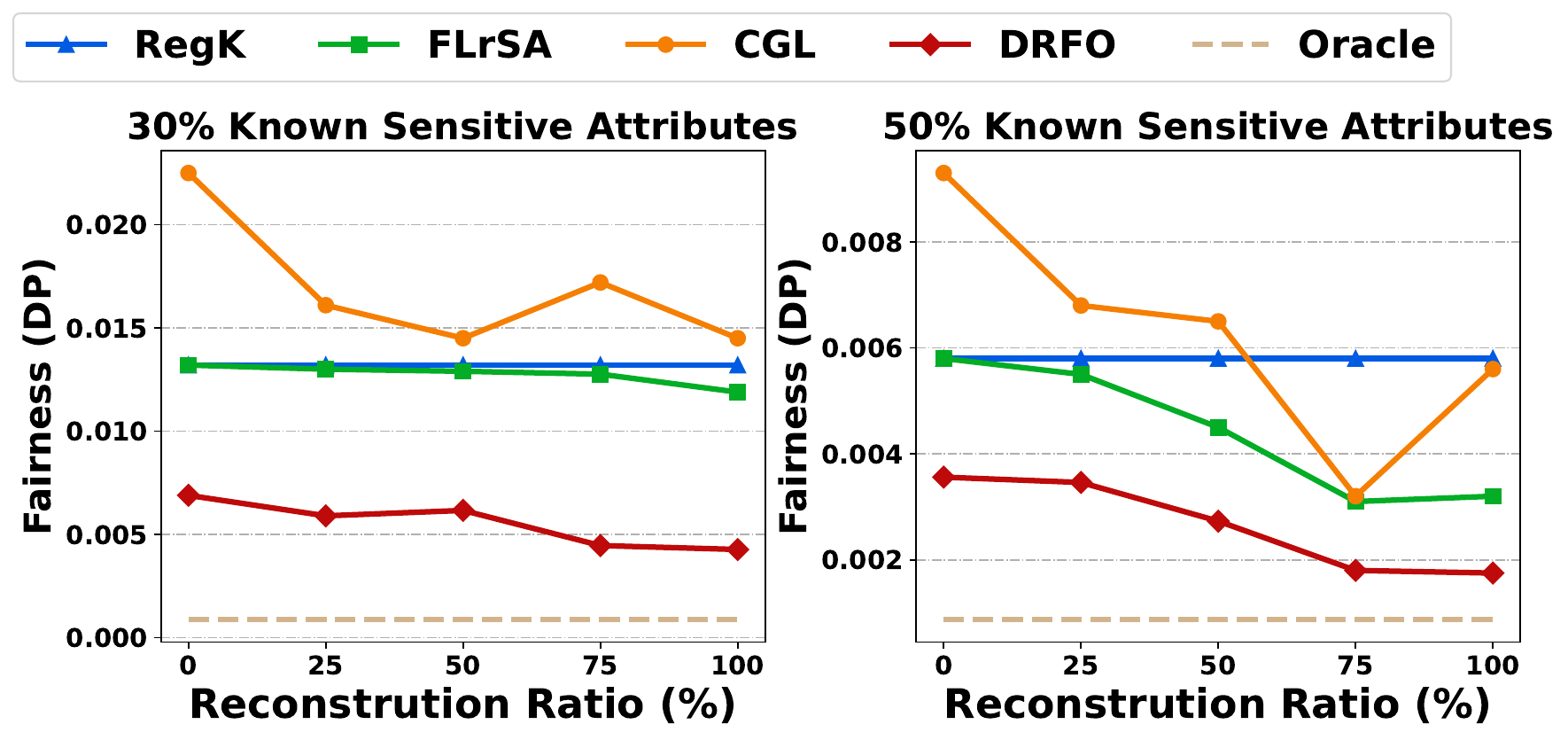}
\vspace{-10pt} % a子标题与a图间距
\end{minipage}
%\caption{fig1}
% \vspace{-0.8cm}
}
% \quad

\vspace{-5pt} %b图与a标题间距
\subfigure[\textbf{Fairness results on Tenrec.}]
 {
\begin{minipage}[b]{0.475\textwidth}
% \vspace{cm}
\centering
\includegraphics[scale=0.28]{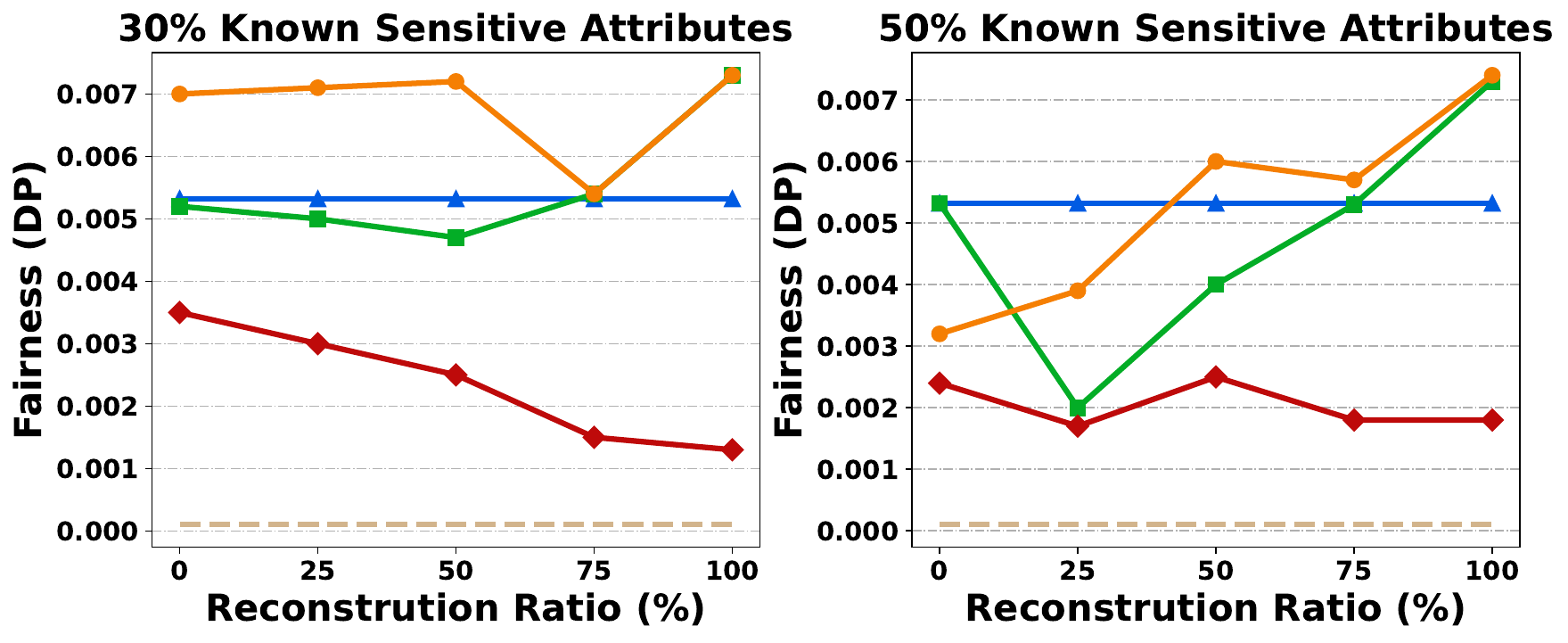}
\vspace{-5pt} % b子标题与图间距
\end{minipage}
}
% \vspace{-0.3cm}
     % \includegraphics[width=0.45\textwidth]{figs/refuse-reconstruction.pdf} 
     \vspace{-10pt} %大标题与上图间距
     \caption{
     Fairness results in scenarios where a portion of users does not allow reconstruction of their attributes among the users with unknown sensitive attributes.
     % 'Ratio' denotes the proportion of users who permit the reconstruction of their sensitive attributes.
     %Comparison of fairness between baselines and DRFO on ML-1M and Tenrec datasets, considering known sensitive attribute ratios of $30\%$ and $50\%$ across various ratios of user consent for sensitive attribute reconstruction (\ie reconstruction sensitive attribute ratio). Lower DP values denote improvements in fairness.
     }
    % \vspace{-0.5cm}
    % \vspace{-15pt}
      \label{fig:refuse-reconstruction}
 \end{figure}
% \vspace{-3pt}
% \subsection{Fairness for Users without Willingness to Reconstruct Sensitive Attributes (RQ3)}
\vspace{-4pt}
\subsection{Fairness for Users Refusing to Reconstruct Sensitive Attributes (RQ3)}
% To further investigate the effectiveness of DRFO in providing fair recommendations for users who do not allow the reconstruction of their sensitive attributes due to privacy concerns, we conducted an additional experiment. In this experiment, we considered scenarios where only a percentage of users (30\% or 50\%) have known sensitive attributes. For the remaining users, we randomly sampled a subset with ratios ranging from 0\% to 100\% to serve as users who allowed the reconstruction of their sensitive attributes, while the rest of the users do not allow the reconstruction.

% To further investigate the effectiveness of DRFO in providing fair recommendations for users who do not allow the reconstruction of their sensitive attributes due to privacy concerns, we conduct additional experiments. In this experiment, we consider scenarios where only a fraction of users (30\% or 50\%) have known sensitive attributes. For the remaining users, we randomly select a subset with ratios ranging from 0\% to 100\% as users who allow the reconstruction of their sensitive attributes, while the rest of the users do not allow the reconstruction.
\sth{To delve deeper into DRFO's ability in ensuring fairness for users refusing sensitive attribute reconstruction due to privacy concerns, we perform supplementary experiments. We explore scenarios where a fraction of users (30\% or 50\%) possess known sensitive attributes. Among the remaining users, we randomly designate a subset, varying from 0\% to 100\%, as individuals permitting sensitive attribute reconstruction, while others opt out of it.}
\sth{ 
To accommodate non-reconstructable
sensitive attributes, we adapt the baselines FLrSA and CGL. For FLrSA, fair constraints are omitted for users not permitting sensitive attribute reconstruction. In the case of CGL, sensitive attributes are randomly assigned to users disallowing reconstruction, mimicking its low-confidence strategy for reconstructed sensitive attributes. A summary of fairness performance across all compared methods is presented in Figure~\ref{fig:refuse-reconstruction}.}
\sth{Figure~\ref{fig:refuse-reconstruction} shows baselines' failures in ensuring fairness when some users do not grant sensitive attribute reconstruction.}
% in scenarios where some users do not allow the reconstruction of their sensitive attributes.
Conversely, DRFO maintains relatively fair results. 
\sth{This validates DRFO's fairness improvement without full sensitive attribute reconstruction by employing a sufficiently large ambiguity set to encompass potential attributes and minimizing the worst-case unfairness within the set.}
\section{Related Work}
\subsection{Fairness in Recommendation}
% In recent years, the expanding reach and influence of recommender systems in people's daily lives have raised significant concerns regarding fairness in recommendations. 
% Recent researches have identified  various forms of discrimination in recommendations,
% The widespread influence of recommender systems has sparked concerns about fairness in recommendations. 
% Recent studies have revealed various forms of discrimination, including bias towards popular item~\cite{abdollahpouri2019unfairness,ge2021towards}, bias influenced by user feedback~\cite{user-orient,causal-fair}, and gender and race discrimination~\cite{chen2018investigating,lambrecht2019algorithmic}.
Fairness research in recommender systems employs two primary paradigms: individual fairness~\cite{wu2021fairness,fairlisa}, which treats similar individuals similarly, and group fairness~\cite{zhang2023chatgpt,recoindependence}, which ensures equitable treatment of different groups.  
% One widely recognized criterion for achieving group fairness is Demographic Parity~\cite{dp,dp2,recoindependence}, which requires decisions to be made similarly regardless of sensitive attributes (\eg gender). 
Our study specifically focuses on user-side fairness in recommendation at a group level, using Demographic Parity~\cite{dp,dp2,recoindependence} as a representative criterion.
% Moreover, considering the diverse stakeholders and their interests, fairness in recommender systems can be examined from both the user perspective and the item perspective~\cite{recfair-survey}. In line with this, our study focuses specifically on user-side fairness at a group level, utilizing demographic parity as a representative criterion.
To promote fairness in recommendations, existing approaches can be broadly categorized into three types~\cite{recfair-survey,li2022fairness}: fair pre-processing, fair learning, and fair adjustment. Fair pre-processing methods mitigate bias in data by resampling~\cite{ekstrand2018all} or adding antidote data~\cite{rastegarpanah2019fighting}. 
Fair learning methods modify optimization targets to learn fair recommendation models, either through fairness criteria as regularizers~\cite{yao2017beyond,recoindependence,tensorfair} or adversarial learning~\cite{wu2021fairness,fairlisa}.
% In-processing methods modify optimization targets to directly learn fair recommendation models. These methods include incorporating fairness criteria as regularizers into the objective function or the reward function~\cite{liu2020balancing}, optimizing models under fairness constraints~\cite{ge2021towards}, and employing adversarial learning to achieve fair representations. 
Fair adjustment methods~\cite{user-orient,wu2021tfrom} reorder the recommendation results to attain a fairer outcome. 
However, these methods typically require full access to sensitive attributes, which may not be feasible in real-world situations. In contrast, our work tackles the challenge of addressing unfairness when only limited sensitive attributes are available.
% To the best of our knowledge, all of these methods require access to all sensitive attributes in order to mitigate discrimination against users belonging to disadvantaged groups, which may be unrealistic in real-world scenarios. In contrast, our work focuses on addressing unfairness when only limited sensitive attributes are available.
% To our knowledge, all these methods require sufficient sensitive attributes to promote fairness, which may be realistic in practice. In contrast, our work focuses on addressing unfairness when only limited sensitive attributes are available.
% \vspace{-0.3cm}
\vspace{-3pt}
\subsection{Fairness with Missing Sensitive Attributes}
% 推荐领域对含有缺失敏感属性的研究不足，但其他领域已经有了一些研究。
% 探索的一个方向是提升rawlsian max-min公平，包括利用dro或arl识别有高分类误差的群体，并提升这部分群体的分类性能。然而，这个方法只能应用与rawlsian max-min公平~cite{}。另一个方向是代理公平，它xxxx，例如，xx直接对与敏感属性高度相关的变量上提升公平性，以提升算法对具有不同敏感属性用户的公平性，然而它的性能高度依赖于代理变量与敏感属性的相关性。

% 与这些方法不同的是，敏感属性重建或敏感表示重建~\cite。例如，fairgnn，cgl。重建的方法具有更高的灵活性，可以针对不同公平指标施加不同公平方法。然而，当重建的敏感属性不够准确时，公平性的提升也受到损害。与之不同的是，我们的方法xxxx。
% Fair recommendation research has placed relatively less emphasis on addressing fairness with missing sensitive attributes, while other domains studying fairness have made some efforts in this regard. 
% Under the assumption of complete unavailability of sensitive attributes, two main directions emerge~\cite{fairness-survey}.
% One direction is to explore Rawlsian Max-Min Fairness~\cite{rawls-fair}, which aims at identifying groups with high classification errors and improve their performance with DRO~\cite{dro-rawlfair} or adversarial reweighted learning~\cite{arl}. Nevertheless, these methods cannot be directly extended to meet fairness criteria such as Demographic Parity~\cite{dp2} or Equal Opportunity~\cite{hardt2016equality}, which require similar classification outcomes for groups with different sensitive attributes. 
% Fair recommendation research has placed relatively less emphasis on addressing fairness with missing sensitive attributes. However, other domains focused on fairness have made some progress in tackling this issue.
Fair recommendation research has paid less attention to fairness when some sensitive attributes are missing, but progress has been made in other fairness-focused domains.
Some methods are
specifically designed to enhance Rawlsian Fairness~\cite{rawls-fair} using DRO~\cite{dro-rawlfair} or Adversarial Reweighting Learning~\cite{arl}, but they cannot extend to other fairness metrics like DP. 
% These approaches are specifically designed for Rawlsian max-min fairness~\cite{zhao2022towards}.
% One direction involves enhancing Rawlsian max-min fairness~\cite{rawls-fair}, which entails utilizing techniques such as DRO~\cite{dro-rawlfair} or Adversarial Reweighting Learning~\cite{arl} to improve the performance of groups with high classification errors.
% to identify groups with high classification errors and improve their classification performance. 
% It is important to note that these approaches are specifically applicable to Rawlsian max-min fairness~\cite{zhao2022towards}.
% However, these approaches are only applicable to Rawlsian max-min fairness~\cite{zhao2022towards}.
% todo rewrite
Another direction utilizes observable attributes (\eg zip code) as proxies for sensitive attributes~\cite{gupta2018proxy,zhao2022towards}, which differs from our method of reconstructing sensitive attributes. The effectiveness of these approaches relies on the strength of the correlation between the sensitive attribute and the observable non-sensitive one.
% Another direction explores the idea that even in the absence of sensitive attributes, other observable attributes (\eg zip code) , which exhibit a high correlation with the sensitive attributes, can be utilized to promote fairness~\cite{gupta2018proxy,zhao2022towards}.
% For example, proxy fairness~\cite{gupta2018proxy} defines proxy groups based on zip codes and formulates a constraint optimization problem for these groups. Nevertheless, the effectiveness of these methods relies on the strength of correlation between the observable attributes and the unobserved sensitive attributes.
% they depends on the strength of correlation between the proxy attributes and the target unobserved sensitive attributes.
% proxy fairness~\cite{gupta2018proxy,zhao2022towards}, whose core idea is that even though sensitive attributes (\eg gender and race) are unobserved, they may be correlated with other observable attributes (\eg zip code). For example, 
% Proxy fairness leverages these observable attributes as proxies to achieve fairness with respect to the unobserved sensitive attributes. The effectiveness of these methods heavily depends on the strength of correlation between the proxy attributes and the target unobserved sensitive attributes~\cite{xxx}. 
% 与这些方法不同的是，一类方法通过重建敏感属性或重建敏感表示然后利用这些重建的信息提升公平性
% The aforementioned methods do not require the acquisition of sensitive attributes, while the following methods assume partial access to sensitive attributes. 
Furthermore, some works enhance fairness by reconstructing missing sensitive attributes~\cite{fairgnn,CGL} or sensitive representations~\cite{Grari2021FairnessWT} and leveraging this reconstructed information. 
% These model-agnostic techniques offer flexibility in applying diverse fairness techniques to address various metrics.
% aims to enhance fairness by reconstructing sensitive attributes~\cite{fairgnn,CGL} or sensitive representations~\cite{Grari2021FairnessWT}, and subsequently leveraging these reconstructed pieces of information.
% For instance, FairGNN~\cite{fairgnn} predicts missing sensitive attributes using graph structures and incorporate them into adversarial learning to improve fairness.
% leverages graph structures and known sensitive attributes to train a classifier that predicts missing sensitive attributes, which are then used in adversarial learning for improving fairness. 
% CGL~\cite{CGL} trains a sensitive attribute classifier with images and assigns random sensitive attributes to images with low prediction confidence to mitigate their impact on fairness.
% For instance, CGL~\cite{CGL} predicts sensitive attributes for images and assigns random sensitive attributes to images with low prediction confidence, reducing their impact on fairness.
% These methods reconstruct sensitive attributes or representations, offering flexibility in applying different fairness techniques for various metrics.
However, they may have limited effectiveness due to reconstruction errors. In contrast, our DRO-based fair learning method ensures fairness in the presence of reconstruction errors. 
% In recommender systems, FairLISA~\cite{fairlisa} studies the problem of missing sensitive attributes. It employs a discriminator trained with known sensitive attributes to eliminate the sensitive information of users whose sensitive attributes are unavailable. 

Within recommender systems, FairLISA~\cite{fairlisa} first studies the problem of missing sensitive attributes. It utilizes a discriminator trained with known attributes to remove the sensitive information from embeddings of users whose sensitive attributes are unavailable. 
% It employs a discriminator trained with known sensitive attributes to eliminate the sensitive information of users whose sensitive attributes are unavailable. 
% FairLISA~\cite{fairlisa} addresses the challenge of missing sensitive attributes by utilizing a discriminator trained with known attributes to remove the sensitive information from embeddings of users whose sensitive attributes are unavailable. 
Unlike it, our method employs the DRO paradigm to offer robust fair recommendations for users with missing sensitive attributes. 
Additionally, FairLISA and our work explore the fairness of limited sensitive attributes from different fairness metrics~\cite{recfair-survey}. While FairLISA focuses on the fairness of removing sensitive information from user embeddings, our focus is on the fairness of inter-group metric differences.
% 我们的方法是
% utilizes reconstructed information to construct an ambiguity set, minimizing worst-case unfairness and ensuring promising fairness outcomes.
% By reconstructing sensitive attributes or sensitive representations, these methods offer the flexibility to apply different fairness techniques for various fairness metrics.
% Nevertheless, these methods may exhibit limited improvements in fairness in the presence of reconstruction errors. In contrast, our proposed approach utilizes the reconstructed information to construct an ambiguity set and ensure promising fairness by minimizing the worst-case unfairness in this ambiguity set.
% 这些方法在存在重建误差的情况下可能对公平的改善有限。与它们不同的是，我们的方法利用重建的信息估计模糊集，并通过最小化模糊集上的最不公平以提供有保障的公平
% 通过重建出敏感属性或敏感表示，这些方法可以灵活地对不同公平指标使用不同公平方法。
% These methods, including ours, involve training sensitive attribute classifiers. However, our method goes further by leveraging distributionally robust optimization to provide theoretical guarantees to ensure fairness with only limited sensitive attributes.

% \vspace{-0.4cm}
\vspace{-3pt}
\subsection{Distributionally Robust Optimization}
% dro:dro是个啥，uncertainty set内优化最差的。
% dro是一个著名框架，解决缺信息。给定一个已知的分布与一个距离度量，dro考虑一个uncertainty set，包含与该已知分布一定距离内的所有分布。
% 已经有有利用dro研究fairness的工作
% todo 或许可以不提dro
Distributionally Robust Optimization (DRO)~\cite{gda,rahimian2019distributionally} is a well-known framework that handles uncertainty. 
It seeks a solution that performs well under all possible distributions within a defined set by optimizing for the worst-case expected loss~\cite{gda,rahimian2019distributionally}.
% This guarantees the solution's reliability even when the true distribution is unknown.
% It works by defining a set of potential distributions based on a known one and a distance measure. DRO aims to find a solution that performs well under all distributions within this set by optimizing for the worst-case expected loss. 
% It defines an ambiguity set based on a known distribution and a distance metric, encompassing distributions within a specific distance from the known one. DRO seeks a robust solution that performs well across all distributions within the set, optimizing the worst-case expected loss~\cite{gda,rahimian2019distributionally}. This guarantees the solution's performance under the unknown true distribution.
% By defining an am
% biguity set based on a known distribution and a distance metric, DRO encompasses all distributions within a certain distance from the known distribution. DRO aims to find a robust solution that performs well across all possible distributions within the set, optimizing the worst-case expected loss~\cite{gda,rahimian2019distributionally}. This ensures the solution's performance on the potential true distribution.
Previous studies have utilized DRO to investigate fairness. Hashimoto et al.~\cite{dro-rawlfair} employ DRO to improve the performance of groups with high classification errors, and Wen et al.~\cite{dro-rec-fair} use DRO to improve worst-case user experience in recommendations. These approaches focus on Rawlsian Fairness~\cite{rawls-fair}, which cannot be extended to other fairness criteria (\eg DP). 
RobFair~\cite{dro-rec-fair} utilizes DRO to provide robust fair recommendations, overcoming the challenge of user preference shift which may compromise the fairness achieved during the training phase. 
Unlike RobFair, our work investigates fair recommendations when some sensitive attributes are missing and uses DRO to handle the uncertainty of missing sensitive attributes.
Additionally, Wang et al.~\cite{wang2020robust} introduce two methods that utilize DRO to learn fair classifiers with noisy sensitive attributes, which is most closely related to our work. In difference, our work considers scenarios where sensitive attributes cannot be reconstructed and discusses how to estimate the upper bound of TV distance when the assumption in Theorem~\ref{theo:tv-distance} is slightly violated.
\vspace{-3pt}
\section{Conclusion}
\sth{In this study, we aim to enhance fairness in recommender systems with limited sensitive attributes.}
We point out that directly reconstructing sensitive attributes may suffer from reconstruction errors and have privacy concerns.
To overcome the challenges, we propose a DRO-based fair learning framework. It builds an ambiguity set based on reconstructed sensitive attributes of users who grant reconstruction, and subsequently optimizes the worst-case unfairness within the entire set.
% Towards this goal, our approach first reconstructs sensitive attributes using available personalized information, and then construct an ambiguity set which encompasses the unknown true distribution based on the reconstructed sensitive attributes, and finally ensure fairness within the entire set using DRO.
% We exhibit theoretical evidence to demonstrate the efficacy of our method in achieving fair recommendations, even in the presence of sensitive attribute reconstruction errors.
We provide theoretical evidence that our methods could build an appropriate ambiguity set that encompasses the unknown true distribution, thus providing robust fair recommendations with limited sensitive attributes.
% to achieve fair recommendations.
% Theoretical evidence showcases our method's efficacy in achieving fair recommendations, even with attribute reconstruction errors.
% Additionally, we conduct experiments on two real-world datasets, providing insightful analyses that validate the effectiveness of our proposal. 
Additionally, extensive experiments on two real-world datasets confirm the efficacy of our methods.
% we validate our method through experiments on two real-world datasets, confirming its efficacy.
% In the future, we will also delve into the exploration of fairness within non-binary sensitive attributes.
In the future, we plan to extend our approach to address fairness in non-binary sensitive attributes. 
Furthermore, our current work only considers using DP as a fairness objective, 
and we aim to extend the approach to other group fairness metrics.
\appendix
% \newpage
\vspace{-2pt}
\section{Appendix}
\subsection{Total Variation Distance}
\label{sec:appendix}
% \subsection{Total Variation Distance}
% \label{sec:appendix-tv}
%%%%%%%%%%%% todo rewrite
In this section, we first present an introduction to the Total Variation (TV) distance, including its definition and some basic properties. Subsequently, we proceed to prove Theorem 2, and from this proof, we will be able to deduce Theorem 1. Afterward, we discuss the situation when the assumptions of Theorem 2 are not satisfied.
% 然后，我们简单讨论一下在定理2假设不满足时候的情况
 %before delving into the proofs of theorem~\ref{theo:tv-distance}.
% \setcounter{definition}{0}
% \begin{definition}
%     00000000
% \end{definition}
% \vspace{-2pt}
\begin{definition}[Total Variation Distance]
\label{def:TVdistance}
The total variation distance between probability measures $P$ and $Q$ on a measurable space $(\Omega, \mathcal{F})$ is defined as the supremum of the absolute difference between the probabilities assigned to a measurable event $A\in\mathcal{F}$ by $P$ and $Q$.
% \begin{equation}\small
% \label{eq:def-tv-original}
%     TV(P,Q)=\sup\limits_{A\in \mathcal{F}} \left| P(A)-Q(A)\right|.
% \end{equation}
If $\Omega$ is countable, the total variation distance can also be expressed in terms of the L1-norm~\cite{TV-distance} as follows:
\begin{equation}\small
\vspace{-2pt}
\label{eq:tv-distance-probability}
    TV(P,Q)=\sup\limits_{A\in \mathcal{F}} \left| P(A)-Q(A)\right|=\frac{1}{2}\left \|P-Q  \right \|_{1}.
    % =\frac{1}{2}\sum_{\omega\in \Omega} \left |P(\left\{\omega \right\}) - Q(\left\{\omega \right\})  \right|.
    \vspace{-2pt}
\end{equation}
\end{definition}

% \begin{definition}
%     111111111
% \end{definition}
% \begin{definition}
%     22222222
% \end{definition}
% A computation example for TV distance can be found in Appendix~\ref{sec:appendix-proof}.

% \subsection{Proofs for Theorem~\ref{theo:tv-distance}}
% \label{sec:appendix-proof}
% In this subsection, we prove Theorem~\ref{theo:tv-distance} using the total variation distance and related concepts.
\setcounter{theorem}{1}
\begin{theorem}
% \vspace{-2pt}
\label{theo:tv-distance2}
Given a measurable event $A$, assuming $P(S=s)=P(\hat{S}=s)$ for a given $s \in \left\{ 0,1 \right \}$, the total variation distance between two conditional probabilities $P(A|S=s)$ and $P(A|\hat{S}=s)$ is bounded by the probability that the sensitive attribute is incorrectly reconstructed, i.e., $TV(P(A|S=s), P(A|\hat{S}=s))\leq P(S\ne \hat{S}|S=s)$.
% Suppose $P(S=s)=P(\hat{S}=s)$ for a given $s \in \left\{ 0,1 \right \}$. Then $TV(p^{(s)},\hat{p}^{(s)})\leq P(S\ne \hat{S}|S=s)$
% \vspace{-2pt}
\end{theorem}

% \begin{proof}
% Given a measurable event $A$ and two conditional probabilities $P(A|S=s)$ and $P(A|\hat{S}=s)$. For simplicity and with a slight abuse of notation, we use $p_{(A|s,\hat{s})}$ to represent $P(A|S=s,\hat{S}=s)$, $p_{(A|\neg s,\hat{s})}$ to represent $P(A|S= s,\hat{S}\ne s)$, $p_{(A| s,\neg \hat{s})}$ to represent $P(A|S\ne s,\hat{S}=s)$, $p_{(A|\neg s,\hat{s})}$ to represent $P(A|S\ne s,\hat{S}=s)$ 
% Then we have:
% \notag
% \begin{small}
% \begin{alignat}{2}
%  % & \ &&\left |p^{(s)}-\hat{p}^{(s)}\right | =
%  & \quad TV(P(A|S=s),P(A|\hat{S}=s)) =  \sup_{A}\left|P(A|S=s )-P(A|\hat{S}=s) \right| \\
%  &=\sup_{A}\left| p_{(A|s,\hat{s})} p_{(\hat{s}|s)} + p_{(A|s,\neg \hat{s})} p_{(\neg\hat{s}|s)}-p_{(A|s,\hat{s})} p_{(s|\hat{s})}-p_{(A|\neg s,\hat{s})} p_{(\neg s|\hat{s})} \right| \\[-10pt]
% &=\sup_{A}\left| p_{(A|s,\hat{s})}\big(p_{(\hat{s}|s)}-p_{(s|\hat{s})} \big) - p_{(\neg\hat{s}|s)} \big(p_{(A|s,\neg \hat{s})}-p_{(A|\neg s,\hat{s})}  \big)
% \right| \\
% &=\sup_{A}\left| 0 - p_{(\neg\hat{s}|s)} \big(
% p_{(A|s,\neg \hat{s})} - p_{(A|\neg s,\hat{s})}\big)
% \right| \\
% &\leq p_{(\neg \hat{s}|s)}=P(\hat{S}\ne s|S=s)=P(\hat{S}\ne S|S=s)
% \end{alignat}
% \end{small}
% \end{proof}

\begin{proof}
 Assuming {\small$P(S\text{=}s)\text{=}P(\hat{S}\text{=}s)$}, the following equation holds:
\notag
\begin{equation*}
\small
P(\hat{S}=s|S=s)=\frac{P(\hat{S}=s,S=s)}{P(S=s)}=\frac{P(S=s,\hat{S}=s)}{P(\hat{S}=s)}=P(S=s|\hat{S}=s),
\end{equation*}
and we can also deduce $P(\hat{S}\ne s|S\text{=}s)\text{=}P(S\ne s|\hat{S}\text{=}s)$. Then, given a measurable event $A$ and two conditional probabilities $P(A|S=s)$ and $P(A|\hat{S}=s)$. We have:
% \notag
% \begin{small}
% \begin{alignat}{2}
%  % & \ &&\left |p^{(s)}-\hat{p}^{(s)}\right | =
%  & \quad TV(P(A|S=s),P(A|\hat{S}=s)) =  \sup_{A}\left|P(A|S=s )-P(A|\hat{S}=s) \right| \\
%  &=\sup_{A}\left| P( A|S=s,\hat{S}=s ) P( \hat{S}=s|S=s ) \text{+} P( A|S=s,\hat{S}\ne s ) P( \hat{S}\ne s|S=s ) \right.\\[-10pt]
%  &\ - \left. P( A|S=s,\hat{S}= s ) P( S=s|\hat{S}=s ) - P( A|S\ne s,\hat{S}=s ) P( S\ne s|\hat{S} = s )\right|. 
%  % \\[-10pt].
%  \end{alignat}
% \end{small}
\notag
\begin{small}
\begin{alignat}{2}
 % & \ &&\left |p^{(s)}-\hat{p}^{(s)}\right | =
 % & \quad TV(P(A|S=s),P(A|\hat{S}=s))\\
 & \quad TV(P(A|S=s),P(A|\hat{S}=s)) =  \sup_{A}\left|P(A|S=s )-P(A|\hat{S}=s) \right| \\[-2pt]
 &=\sup_{A}\left| P( A|S=s,\hat{S}=s ) P( \hat{S}=s|S=s ) \text{+} P( A|S=s,\hat{S}\ne s ) P( \hat{S}\ne s|S=s ) \right.\\[-10pt]
 &\ - \left. P( A|S=s,\hat{S}= s ) P( S=s|\hat{S}=s ) - P( A|S\ne s,\hat{S}=s ) P( S\ne s|\hat{S} = s )\right| \\[-10pt]
 % &=\sup_{A}\left| P( A|S=s,\hat{S}=s ) P( \hat{S}=s|S=s ) \text{+} P( A|S=s,\hat{S}\ne s ) P( \hat{S}\ne s|S=s ) \right.\\[-10pt]
 % &\ - \left. P( A|S=s,\hat{S}= s ) P( S=s|\hat{S}=s ) - P( A|S\ne s,\hat{S}=s ) P( S\ne s|\hat{S} = s )\right| \\[-5pt]
&=\sup_{A}\left|P(A|S=s,\hat{S}=s)\left( P(\hat{S}=s|S=s) - P(S=s|\hat{S}=s)\right) \right.\\[-5pt]
&\ -\left.P(\hat{S}\ne S|S=s)\left(P(A|S=s,\hat{S}\ne s)-P(A|\hat{S}=s,S\ne s) \right)\right|\\[-2pt]
&=\sup_{A}\left| 0 - P(\hat{S}\ne S|S=s) \left(P(A|S=s,\hat{S}\ne s) - P(A|\hat{S}=s,S\ne s) \right)\right| \\[-10pt]
&\leq P(\hat{S}\ne S|S=s).
\qquad \qquad \qquad \qquad \qquad \qquad \qquad \qquad \qquad  \qedhere
 \end{alignat}
\end{small}
\end{proof}
\vspace{-3pt}
The above derivation completes the proof of Theorem~\ref{theo:tv-distance2}, thereby deriving Theorem~\ref{theo:tv-distance}. Note that we assume $P(S)=P(\hat{S})$ in proving Theorem~\ref{theo:tv-distance2}. 
If the assumption is violated, as the general reconstructive ability to accurately restore true sensitive attributes, we presume a slight deviation between $P(S)$ and $P(\hat{S})$, that is, $P(\hat{S}) - P(S)  = \delta p (\left|\delta p\right|\ll P(S))$, then
% If the assumption is violated, 
% 这就完成了定理二的证明，即假设P（S）=P（S）时，从而可以推导出定理一。当然，这假设P（S）=P（S）。当假设不成立时，由于一般重建的敏感属性可以比较准确地还原出真实敏感属性，我们假定，P（S）与P（、hat{S}）中存在小的偏差，即xxx，那么
% ，xxx，从而可以推导出定理一。定理二假设P（S）=P(\hat{S})。当假设不成立
% In our experiments, due to the relatively high accuracy in reconstructing sensitive attributes, $P(S)$ and $P(\hat{S})$ are close, so we assume $P(S)=P(\hat{S})$. If we consider the existence of a small deviation, that is, $ P(\hat{S}) - P(S)  = \delta p (\left|\delta p\right|\ll P(S))$, then
\notag
\begin{small}
\begin{alignat}{2}
 % & \ &&\left |p^{(s)}-\hat{p}^{(s)}\right | =
 & \quad TV(P(A|S=s),P(A|\hat{S}=s))\\
  % &=\sup_{A}\left| P( A|S=s,\hat{S}=s ) P( \hat{S}=s|S=s ) \text{+} P( A|S=s,\hat{S}\ne s ) P( \hat{S}\ne s|S=s ) \right.\\[-10pt]
 % &\ - \left. P( A|S=s,\hat{S}= s ) P( S=s|\hat{S}=s ) - P( A|S\ne s,\hat{S}=s ) P( S\ne s|\hat{S} = s )\right| \\[-5pt]
&=\sup_{A}\left|P(A|S=s,\hat{S}=s)\left( P(\hat{S}=s|S=s) - P(S=s|\hat{S}=s)\right) \right.\\[-2pt]
& - \left. \left( P(A|S\ne s, \hat{S}=s)P(S\ne s|\hat{S}=s)- P(A|S=s,\hat{S}\ne s)P(\hat{S}\ne s|S=s) \right )\right |
\\[-10pt]
& 
\approx \sup_{A} \left | \left (kP(A|S=s,\hat{S}=s)\delta p -  kP(A|S=s,\hat{S}\ne s)\delta p \right )  \right.\\[-3pt]
&-\left.
P(\hat{S}\ne S|S=s) \left(P(A|S=s,\hat{S}\ne s) - P(A|\hat{S}=s,S\ne s) \right) 
\right|
\\
% -\left.P(\hat{S}\ne S|S=s)\left(P(A|S=s,\hat{S}\ne s)-P(A|\hat{S}=s,S\ne s) \right)\right|\\[-2pt]
% &=\sup_{A}\left| 0 - P(\hat{S}\ne S|S=s) \left(P(A|S=s,\hat{S}\ne s) - P(A|\hat{S}=s.S\ne s) \right)\right| \\[-8pt]
&\leq P(\hat{S}\ne S|S=s) + k\left|\delta p\right |,
% \qquad \qquad \qquad \qquad \qquad \qquad \qquad \qquad \qquad  \qedhere
 \end{alignat}
\end{small}
where $k=\frac{P(S=s,\hat{S}=s)}{P(S=s)P(\hat{S}=s)}$, and the approximate equality is obtained through the probability formula and first-order approximation (details omitted for simplicity). 
From the derivation, we know that when there is a marginal discrepancy between $P(S)$ and $P(\hat{S})$, the estimated upper bound on the TV distance should be looser. However, 
% this leniency is of modest magnitude, 
the impact on the upper bound is modest, 
rendering its omission reasonable. Also, in our experiments, using only the first term (\ie $P(\hat{S}\ne S |S=s)$) to build ambiguity sets has proven to achieve robust fairness, implying the rationale of directly utilizing the conclusion of Theorem~\ref{theo:tv-distance2} to estimate the upper bound of the TV distance.

\begin{acks}
This work is supported by the National Key Research and Development Program of China (2022YFB3104701), the National Natural Science Foundation of China (62272437), and the CCCD Key Lab of Ministry of Culture and Tourism.
\end{acks}

\bibliographystyle{ACM-Reference-Format}

%\bibliography{sample-base}
% \newpage
\bibliography{ref-bib}

\end{document}